\documentclass[12pt,reqno]{amsart}

\usepackage{times}
\usepackage{graphicx}
\graphicspath{ {./images/} }

\usepackage{amsmath,amsthm,amsfonts,amscd,amssymb}
\usepackage{pst-node}
\usepackage{pst-plot}
\usepackage[all]{xy}
\usepackage{stackengine}

\usepackage{tikz, braids}
\usetikzlibrary{decorations.pathreplacing, arrows}

\newtheorem{theorem}{Theorem}[section]
\newtheorem{corollary}[theorem]{Corollary}

\newtheorem{lemma}[theorem]{Lemma}
\newtheorem{remark}[theorem]{Remark}
\newtheorem{proposition}[theorem]{Proposition}

\numberwithin{theorem}{section}
\numberwithin{equation}{section}

\newcommand{\Real}{\mathbb R}

\newcommand{\set}[1]{\left\{#1\right\}}

\newcommand{\la}{\langle}
\newcommand{\ra}{\rangle}

\newcommand{\Comp}{\mathbb{C}}

\newcommand{\D}{\mathcal{D}}
\newcommand{\Hi}{\mathcal{H}}

\newcommand{\Le}{\mathcal{L}}
\newcommand{\M}{\mathcal{M}}

\newcommand{\Om}{\Omega}
\newcommand{\om}{\omega}

\global\long\def\tp{\mathop{\xymatrix{*+<.7ex>[o][F-]{\scriptstyle \top}}
 } }

\begin{document}

\title[Entanglement detection under group symmetry]{A universal framework for entanglement detection under group symmetry}
% or if you want, simply \title{Title of the article}

%\author{Author name}

\author[S.-J. Park]{Sang-Jun Park}
\address{Sang-Jun Park, 
Department of Mathematical Sciences, Seoul National University, 
Gwanak-Ro 1, Gwanak-Gu, Seoul 08826, Republic of Korea}
\email{psj05071@snu.ac.kr}

\author[Y.-G. Jung]{Yeong-Gwang Jung}
\address{Yeong-Gwang Jung, 
Department of Mathematics Education, Seoul National University, 
Gwanak-Ro 1, Gwanak-Gu, Seoul 08826, Republic of Korea}
\email{wollow21@snu.ac.kr}

\author[J. Park]{Jeongeun Park}
\address{Jeongeun Park, 
Department of Mathematics Education, Seoul National University,
Gwanak-Ro 1, Gwanak-Gu, Seoul 08826, Republic of Korea}
\email{pju0013@snu.ac.kr}

\author[S.-G. Youn]{Sang-Gyun Youn}
\address{Sang-Gyun Youn, 
Department of Mathematics Education, Seoul National University, 
Gwanak-Ro 1, Gwanak-Gu, Seoul 08826, Republic of Korea}
\email{s.youn@snu.ac.kr }

\maketitle

\begin{abstract}
One of the most fundamental questions in quantum information theory is PPT-entanglement of quantum states, which is an NP-hard problem in general. In this paper, however, we prove that all PPT $(\overline{\pi}_A\otimes \pi_B)$-invariant quantum states are separable if and only if all extremal unital positive $(\pi_B,\pi_A)$-covariant maps are decomposable where $\pi_A,\pi_B$ are unitary representations of a compact group and $\pi_A$ is irreducible. Moreover, an extremal unital positive $(\pi_B,\pi_A)$-covariant map $\Le$ is decomposable if and only if $\Le$ is completely positive or completely copositive. We then apply these results to prove that all PPT quantum channels of the form $$\Phi(\rho)=a\frac{\text{Tr}(\rho)}{d}\text{Id}_d+ b\rho+c\rho^T+(1-a-b-c)\text{diag}(\rho)$$ are entanglement-breaking, and that all A-BC PPT $(U\otimes \overline{U}\otimes U)$-invariant tripartite quantum states are A-BC separable. The former strengthens some conclusions in \cite{VW01,KMS20}, and the latter provides a strong contrast to the fact that there exist PPT-entangled $(U\otimes U\otimes U)$-invariant tripartite Werner states \cite{EW01} and resolves some open questions raised in \cite{COS18}.
\end{abstract}

\section{Introduction}

Quantum entanglement is one of the most non-classical manifestations of quantum formalism and is considered a key resource for quantum communication. Indeed, quantum entanglement plays crucial roles in the existence of Bell correlations \cite{Be64,We89}, quantum cryptography \cite{Ek91,JSWWZ00,TBZG00,NPWBK00}, superdense coding \cite{BeWi92,MWKZ96}, quantum teleportation \cite{BBC93,BPMEWZ}, entanglement-assisted classical communication \cite{BSSJT}, and computational supremacy for communication complexity problems \cite{Br03,BvDHT,CDD}.

Whether a given bipartite quantum state is entangled is a fundamental question in quantum information theory(QIT). While the above question is NP-hard in general \cite{Gu03,Gh10}, there is one effective test for entanglement called the PPT (positive partial transpose) criterion, saying that non-PPT states are entangled. The PPT criterion leaves us questioning whether PPT states can be entangled. It turns out that PPT entangled states are bound entangled, which are applicable to perform non-classical tasks \cite{HHH99,VW02,Ma06} and to produce secure cryptographic key \cite{HHHO05,HHHO09,HPHH08}. Note that there are numerous results on multipartite entanglements including the classification of entanglements for GHZ states in various contexts \cite{Ka11,Gu11,HK16,HaKy16}. %However, we will focus on bipartite entanglement in this paper.

Given the difficulty of testing entanglement, it is a common technique in quantum theory to impose certain symmetry in the hope that the symmetry allows us to restrict our attention to more tractable models. In this paper, we focus on {\it invariant quantum states} with respect to {\it compact} group symmetries. More precisely, for finite-dimensional unitary representations $\pi_A:G\rightarrow B(H_A)$ and $\pi_B:G\rightarrow B(H_B)$ of a compact group $G$, a bipartite quantum state $\rho\in \mathcal{D}(H_A\otimes H_B)$ is called $(\pi_A\otimes \pi_B)$-invariant if
\begin{equation}
(\pi_A(x)\otimes \pi_B(x))\rho = \rho (\pi_A(x)\otimes \pi_B(x))
\end{equation}
for all $x\in G$. Werner states and isotropic states are standard examples of invariant quantum states with respect to the standard representations of unitary groups, and their entanglements were characterized in \cite{We89} and \cite{HoHo99}, respectively. The cases of other group symmetries can be found in \cite{VW01, EW01,KCL05,DPR07b,TG09,Al14,BCS20,SiNe21,CKKLY21}.

The dual objects of invariant quantum states are the so-called {\it$(\overline{\pi_A},\pi_B)$-covariant quantum channels}, which are completely positive trace-preserving (CPTP) maps $\mathcal{L}:B(H_A)\rightarrow B(H_B)$ satisfying
\begin{equation}
\mathcal{L}(\overline{\pi_A(x)}X\pi_A(x)^T)=\pi_B(x)\mathcal{L}(X)\pi_B(x)^*
\end{equation}
for all $X\in B(H_B)$ and $x\in G$. Indeed, for a linear map $\mathcal{L}:B(H_A)\rightarrow B(H_B)$ and its normalized Choi matrix $C_{\mathcal{L}}= \frac{1}{d_A}\sum_{i,j=1}e_{ij}\otimes \mathcal{L}(e_{ij})$, the given map $\mathcal{L}$ is $(\overline{\pi_A},\pi_B)$-covariant if and only if $C_{\mathcal{L}}$ is $(\pi_A\otimes \pi_B)$-invariant (Corollary \ref{cor-twirling}). Furthermore, a $(\overline{\pi_A},\pi_B)$-covariant quantum channel $\mathcal{L}$ is PPT non-EB(entanglement breaking) if and only if the Choi matrix $C_{\mathcal{L}}$ is a $(\pi_A\otimes \pi_B)$-invariant PPT entangled quantum state from the so-called {\it channel-state duality.}
 
Apart from the channel-state duality, there is another dual picture to understand PPT entanglement via the so-called {\it trace duality}. The well-known Horodecki criterion from \cite{HHH96} is that a bipartite quantum state $\rho\in \mathcal{D}(H_A\otimes H_B)$ is entangled if and only if $(\text{id}\otimes \mathcal{L})(\rho)\ngeq 0$ for some positive linear map $\mathcal{L}:B(H_B)\rightarrow B(H_A)$. This says that positive maps can be regarded as detectors for quantum entanglement. Moreover, if we want to witness PPT entangled states, $\Le$ should be taken to be \textit{non-decomposable} in the sense that $\mathcal{L}$ cannot be written as a sum of CP maps and CCP (completely copositive) maps. However, one technical difficulty is that all positive maps from $B(H_B)$ into $B(H_A)$ should be in our consideration, and verifying whether a linear map is positive is computationally intractable \cite{La10,NiZh16}.

One crucial aspect of this paper is that our abstract approach using group symmetries offers additional advantages on the Horodecki criterion. Indeed, it is enough to consider only positive $(\pi_B,\pi_A)$-covariant maps $\mathcal{L}$ to analyze entanglement of $(\overline{\pi_A}\otimes \pi_B)$-invariant quantum states $\rho$ (Theorem \ref{cor-EBT}). Under a mild assumption of irreducibility of $\pi_A$, we can further restrict $\Le$ to be unital, and we may even choose $\Le$ to be extremal (Theorem \ref{thm-main}). {In other words, these extremal maps (in the set of unital positive covariant maps) are sufficient to detect all entangled invariant states. This not only substantially reduces the complexity on the Horodecki criterion, but also establishes a general framework for duality between separable states and positive maps under group symmetry.}

Another key aspect is to provide a unified perspective to understand the question of PPT entanglement via the channel-state duality and the trace duality under our framework. More precisely, given the irreducibility of $\pi_A$, we can prove that the following three statements are equivalent (Corollary \ref{cor-equiv}):

\begin{itemize}
        \item All PPT $(\overline{\pi_A}\otimes \pi_B)$-invariant quantum states are separable (PPT=SEP).
        \item All PPT $(\pi_A,\pi_B)$-covariant quantum channels are entanglement-breaking (PPT=EB).
        \item All (extremal) unital positive $(\pi_B,\pi_A)$-covariant maps are decomposable (POS=DEC).
\end{itemize}
Thus, the question of PPT entanglement on the set of $(\overline{\pi_A}\otimes \pi_B)$-invariant quantum states (or on the set of all $(\pi_A,\pi_B)$-covariant quantum channels) transfers to non-decomposability on the extremal unital positive $(\pi_B,\pi_A)$-covariant maps.
%{\color{blue} At the same time, this equivalence allows us to unify two independent efforts to find PPT entangled states and positive non-decomposable maps whenever we consider certain symmetries.} 
Moreover, an extremal unital positive $(\pi_B,\pi_A)$-covariant map $\Le$ is decomposable if and only if $\Le$ is completely positive(CP) and completely copositive(CCP) by Theorem \ref{thm-extremal-decomposable}, and both the conditions CP and CCP are much easier to verify.

The merits of our suggested framework can be demonstrated with examples in Section \ref{sec-examples1} and \ref{sec-examples2}. In Section \ref{sec-examples1}, we prove that PPT=EB holds for any quantum channels of the form
\begin{equation}\label{eq10}
    \Phi(\rho)=a\frac{\text{Tr}(\rho)}{d}\text{Id}_d+ b\rho+c\rho^T+(1-a-b-c)\text{diag}(\rho),
\end{equation}
where $\displaystyle \text{diag}(X)=\sum_{i=1}^d X_{ii}e_{ii}$ for $X=(X_{ij})_{1\leq i,j\leq d}$. These channels are precisely the covariant channels with respect to the standard representation of the {\it signed symmetric group} (or the \textit{hyperoctahedral group}) $\mathcal{H}_d$ (Proposition \ref{lem20}). Our results strengthen some conclusions from Section 5 and 6 of \cite{KMS20} and Example 3 of \cite{VW01} to a larger class, and resolve some of open questions posed in \cite{KMS20}. {We refer Remark \ref{rmk-VW01}, \ref{rmk-KMS20}, and Appendix \ref{sec-CovMU} for details.}

In Section \ref{sec-examples2}, we focus on the question PPT=SEP for some tripartite quantum states with unitary group symmetries. In Section 5.1, we present explicit positive non-decomposable covariant linear maps $\Le:M_d(\Comp)\rightarrow M_{d^2}(\Comp)$ satisfying the covariance:
\begin{equation}
\Le(\overline{U}X U^T)=(U\otimes U)\Le(X)(U\otimes U)^*
\end{equation}
for all $d\times d$ unitary matrices $U\in \mathcal{U}_d$ and $X\in M_d(\Comp)$. This result gives another explanation of the fact PPT$\neq$SEP for {\it tripartite Werner states} \cite{EW01}, which implies the existence of PPT entangled quantum states $\rho\in M_{d^3}(\Comp)$ satisfying the following invariance,
\begin{equation}
(U\otimes U\otimes U)\rho = \rho (U\otimes U\otimes U)
\end{equation}
for all $U\in \mathcal{U}_d$. On the other hand, in Section \ref{sec:Q.orthogonal}, we show that a strong contrast PPT$=$SEP holds for {\it quantum orthogonally invariant} quantum states. More generally, we prove that any PPT tripartite quantum state $\rho\in M_{d^3}(\Comp)$ satisfying the following invariance
\begin{equation}
(U\otimes \overline{U}\otimes U)\rho = \rho (U\otimes \overline{U}\otimes U)
\end{equation}
for all unitary matrices $U\in \mathcal{U}_d$ is separable (Theorem \ref{thm-qOOO}). This resolves some open questions raised in \cite{COS18} {(see Remark \ref{rmk-COS18} for details)}.

%\small
%\begin{table}[h!] \label{tab-1}
 % \begin{center}
  %  \caption{A summary of examples}
   % \label{tab:table1}
    %\begin{tabular}{|c|c|c|c|c|c|c|} % <-- Alignments: 1st column left, 2nd middle and 3rd right, with vertical lines in between
       %$\backslash$ 
%\hline
%Phase spaces &Symmetry& SEP=PPT & $\sharp$ (detectors) & Reference\\
 %       \hline
%$\mathcal{U}_d$& U$\otimes $U  & True & 2 & \cite{We89}\\
 %       \hline
%$\mathcal{U}_d$& U$\otimes \overline{\text{U}}$  & True & 2&\cite{HoHo99} \\
%\hline
%$\mathcal{SU}(2)$ & $\pi_1\otimes \overline{\pi_l}$ & True &2& \cite{CKKLY21} \\
%\hline
%$\mathcal{SU}(2)$ & $\pi_2\otimes \overline{\pi_2}$ & True &4&  \cite{CKKLY21}\\
%\hline
%$\mathcal{O}_d$ & $ \text{O}\otimes \text{O}$ & True &4& New - Section 4.4\\
%\hline
%$\mathcal{H}_{d}$ ($d\geq 3$)& $\text{H}\otimes \text{H}$ & True &8& New - Section 4.4 \\
%\hline
%$\mathcal{U}_d$ ($d=2$) & U$\otimes $U$\otimes $U & True &$\infty$& \cite{EW01} \\
%\hline
%$\mathcal{U}_d$ ($d\geq 3$) & U$\otimes $U$\otimes $U & False &$\infty$&\cite{EW01} \\
%\hline
%$\tor^d$ ($d=2$) & D$\otimes \overline{\text{D}}$ &$\infty$& True\\
%\hline
%$\tor^d$ ($d\geq 3$) & D$\otimes \overline{\text{D}}$ &$\infty$& False\\
%\hline
%$\mathcal{O}_d^+$ & $\text{O}^+\otimes \text{O}^+\otimes \text{O}^+$ & False & $\infty$  &New - Section 5.2\\
%\hline
%    \end{tabular}
%  \end{center}
%\end{table}

\normalsize

\section{Preliminaries} \label{sec-pre}

%{\color{red}Notation: $B(H)$, $M_d(\C), {\rm id}_A, {\rm Id}_H$,...??}\\\\

\subsection{Separability and PPT property}

In this paper, we focus only on finite-dimensional complex Hilbert spaces $H=\Comp^d$, $H_A=\Comp^{d_A}$, $H_B=\Comp^{d_B}$, and their direct sums and tensor products. Recall that a {\it quantum state} $\rho\in B(H)$ is a positive matrix with $\text{Tr}(\rho)=1$ and the set of all quantum states in $B(H)$ is denoted by $\mathcal{D}(H)$. A bipartite positive operator $X\in B(H_A\otimes H_B)$ is said to be of {\it positive partial transpose (PPT)} if 
\begin{equation}\label{eq-PPT}
(\text{id}_A\otimes T_B)(X)\geq 0
\end{equation}
where $T_B$ is the transpose map on $B(H_B)$, and $X$ is called {\it separable} if there exist families of positive operators $(X^A_i)_{i=1}^n$ and $(X^B_i )_{i=1}^n$ such that
\begin{equation}
    X=\sum_{i=1}^n X^A_i\otimes X^B_i.
\end{equation}
In particular, if $\rho\in \mathcal{D}(H_A\otimes H_B)$ is a separable quantum state, then there exists a probability distribution $(p_i)_{i=1}^n$ and a family of product quantum states $(\rho_i^A\otimes \rho^B_i)_{i=1}^n$ such that 
\begin{equation}\label{eq-separability}
\rho=\sum_{i=1}^n p_i \rho^A_i\otimes \rho^B_i.
\end{equation}
It is clear that separability implies PPT property, but the converse is not true in general. More precisely, all PPT quantum states in $B(H_A\otimes H_B)$ are separable if and only if $d_A\cdot d_B \leq 6$ \cite{Pe96,HHH96,Wo76,Ch80}. Moreover, it is known that the separability question is NP-hard \cite{Gu03,Gh10}.

For $v\in H$, we define linear maps $|v\ra:\Comp\rightarrow H$ given by $\lambda\mapsto \lambda v$ and $\la v|: H\rightarrow \Comp$ given by $w\mapsto \la v | w\ra$ where $\la v | w\ra$ is the inner product of $v,w\in H$ whose first variable is the anti-linear part. In particular, $|\Omega\ra=\sum_{i=1}^d\frac{1}{\sqrt{d}}|i\ra\otimes |i\ra \in H\otimes H $ is called the \textit{maximally entangled Bell state} where $\set{|1\ra,|2\ra, \cdots, |d\ra}$ is the standard orthonormal basis of $H$. The matrix unit $|i\ra\la j|$ and the product vector $|i_1\ra\otimes |i_2\ra\otimes \cdots \otimes |i_k\ra$ are also denoted by $e_{ij}$ and $|i_1i_2\cdots i_k \ra$ respectively.

The \textit{(normalized) Choi matrix} of a linear map $\mathcal{L}:B(H_A)\rightarrow B(H_B)$ is defined by
\begin{align}
C_{\mathcal{L}}&=(\text{id}_A\otimes \Le)(|\Omega_A\ra\la \Omega_A|)=(\text{id}_A\otimes \Le)\left(\frac{1}{d_A}\sum_{i,j=1}^{d_A}e_{ij}\otimes e_{ij}\right) \nonumber \\
&=\frac{1}{d_A}\sum_{i,j=1}^{d_A}e_{ij}\otimes \Le(e_{ij})\in B(H_A\otimes H_B).
\end{align}

Recall that $\mathcal{L}$ is {\it completely positive} (CP) if and only if the Choi matrix $C_{\mathcal{L}}$ is positive, and $\mathcal{L}$ is {\it trace-preserving} (TP) if and only if $(\text{id}_A\otimes \text{Tr}_B)(C_{\mathcal{L}})=\displaystyle \frac{1}{d_A}\text{Id}_A$. In particular, if $\Phi:B(H_A)
\to B(H_B)$ is a CPTP linear map, i.e. a {\it quantum channel} in the Schr{\" o}dinger's picture, then the Choi matrix $C_{\Phi}$ is a quantum state in $\mathcal{D}(H_A\otimes H_B)$. We call this {\it channel-state duality}. 

Let $\mathcal{L}:B(H_A)\rightarrow B(H_B)$ be a linear map. Then $\mathcal{L}$ is called {\it completely copositive} (CCP) if $T_B\circ \mathcal{L}$ is completely positive, $\mathcal{L}$ is called {\it decomposable} if there exist a CP map $\Le_1$ and a CCP map $\Le_2$ such that $\Le=\Le_1+\Le_2$, and $\Le$ is called PPT if $\Le$ is both CP and CCP. Thus, $\Le$ is PPT if and only if $C_{\Le}$ is PPT.

Another important property of quantum channels is the {\it entanglement-breaking (EB) property} \cite{HSR03}. A quantum channel $\Phi:B(H_A)\rightarrow B(H_B)$ is called EB if the Choi matrix $C_{\Phi}$ is a separable quantum state. Note that any EB quantum channel is PPT, but the converse is not true in general.

\subsection{Invariance and covariance}\label{sec-rep.theory}

In this section, we introduce two important objects to discuss conservation of symmetry, namely {\it invariant operators} and {\it covariant linear maps}. Let us suppose that $G$ is a compact group throughout this paper. A continuous function $\pi:G\rightarrow \mathcal{U}(H_{\pi})$ is called a (finite-dimensional) {\it unitary representation} of $G$ if it is a group homomorphism, i.e., $\pi(xy)=\pi(x)\pi(y)$ for all $x,y\in G$. In this case, an operator $X\in B(H_{\pi})$ is called \textit{$\pi$-invariant} if 
\begin{equation}
\pi(x)X\pi(x)^*=X
\end{equation} 
for all $x\in G$. The set of all $\pi$-invariant operators, the set of all $\pi$-invariant quantum states, and the set of $\pi$-invariant PPT quantum states in $B(H_{\pi})$ are denoted by $\text{Inv}(\pi)$, $\text{InvQS}(\pi)$, and $\text{InvPPTQS}(\pi)$, respectively. 

A unitary representation $\pi:G\rightarrow B(H_{\pi})$ is called {\it irreducible} if $\text{Inv}(\pi)=\Comp\cdot \text{Id}_{H_{\pi}}$. If $\pi$ is irreducible, so is the \textit{contragredient representation} $\overline{\pi}:G\to \mathcal{U}(H_{\pi})$ of $\pi$ which is defined by $\overline{\pi}(x)=\overline{\pi(x)}$ for all $x\in G$. For unitary representations $\pi_A:G\rightarrow \mathcal{U}(H_A)$ and $\pi_B:G\rightarrow \mathcal{U}(H_B)$, the {\it tensor representation} $\pi_A\otimes \pi_B:G\rightarrow \mathcal{U}(H_A\otimes H_B)$ is given by 
\begin{equation}
(\pi_A\otimes \pi_B)(x)=\pi_A(x)\otimes \pi_B(x)
\end{equation}
for all $x\in G$. 

For unitary representations $\pi_A:G\rightarrow B(H_A)$ and $\pi_B:G\rightarrow B(H_B)$, a linear map $\Le:B(H_A)\rightarrow B(H_B)$ is called {\it $(\pi_A,\pi_B)$-covariant} if 
\begin{equation}
\Le(\pi_A(x)Y\pi_A(x)^*)=\pi_B(x)\Le(Y)\pi_B(x)^*
\end{equation}
for all $x\in G$ and $Y\in B(H_A)$,
%Furthermore, if $\pi_A$ and $\pi_B$ are irreducible, then $\Le$ is called {\it irreducibly $(\pi_A,\pi_B)$-covariant}.
and let us denote by $\text{Cov}(\pi_A,\pi_B)$ the space of all $(\pi_A,\pi_B)$-covariant linear maps. Some subclasses of $\text{Cov}(\pi_A,\pi_B)$ in our interest are as follows: 
\begin{itemize}
\item $\text{CovPos}(\pi_A,\pi_B)$ is the set of all $(\pi_A,\pi_B)$-covariant positive maps,
\item $\text{CovPos}_1(\pi_A,\pi_B)$ is the set of all $(\pi_A,\pi_B)$-covariant positive unital maps,
\item $\text{CovPosTP}(\pi_A,\pi_B)$ is the set of all $(\pi_A,\pi_B)$-covariant positive TP maps,
\item $\text{CovQC}(\pi_A,\pi_B)$ is the set of all $(\pi_A,\pi_B)$-covariant CPTP maps,
\item $\text{CovPPTQC}(\pi_A,\pi_B)$ is the set of all $(\pi_A,\pi_B)$-covariant PPT quantum channels.
\end{itemize}

\subsection{Twirling operation} \label{sec-twirling}

An averaging technique called the \textit{twirling operation} is a standard method to analyze invariant operators and covariant linear maps. For a unitary representation $\pi:G\to \mathcal{U}(H)$, we define a  twirling map $\mathcal{T_\pi}:B(H)\rightarrow \text{Inv}(\pi)$ by
\begin{equation}
    \mathcal{T}_{\pi}(X)=\int_G \pi(x)X\pi(x)^* dx
\end{equation}
for all $X\in B(H)$, where $dx$ denotes the \textit{normalized Haar measure} on $G$. Then $\mathcal{T}_{\pi}$ is unital CPTP, and its well-definedness, i.e. $\mathcal{T}_{\pi}(X)\in \text{Inv}(\pi)$, is thanks to the translation-invariance property of the Haar measure. Furthermore, we have $X\in \text{Inv}(\pi)$ if and only if $\mathcal{T}_{\pi}(X)=X$,
%(necessity is clear, and for sufficiency we again use the translation-invariance of the Haar measure). 
which means that $\mathcal{T}_{\pi}$ is a projection (more precisely, a \textit{conditional expectation}) onto the $*$-subalgebra $\text{Inv}(\pi)$ of $B(H)$.

For unitary representations $\pi_A:G\to \mathcal{U}(H_A)$ and $\pi_B:G\to \mathcal{U}(H_B)$, the twirling $\mathcal{T}_{\pi_A,\pi_B}\Le$ of $\mathcal{L}:B(H_A)\rightarrow B(H_B)$ is defined by
\begin{equation}
    (\mathcal{T}_{\pi_A,\pi_B}\Le)(X)=\int_G \pi_B(x)^*\Le(\pi_A(x)X\pi_A(x)^*)\pi_B(x)\,dx\
\end{equation}
for all $X\in B(H_A)$. Then similarly, the twirling operation $\mathcal{T}_{\pi_A,\pi_B}$ is a well-defined projection from $B(B(H_A),B(H_B))$ onto $\text{Cov}(\pi_A,\pi_B)$.
%, and trace-preserving condition, positivity and complete positivity are preserved via $\mathcal{T}_{\pi_A,\pi_B}$.

Let us collect some useful properties of the twirling operations for the next section.

\begin{proposition} \label{prop-twirling-preserving}
For any unitary representations $\pi_A$ and $\pi_B$ of $G$, the twirling map $\mathcal{T}_{\pi_A\otimes \pi_B}$ preserves separability and PPT property of bipartite operators. Furthermore, the twirling operation $\mathcal{T}_{\pi_A,\pi_B}$ preserves positivity, CP, TP, CCP, PPT, decomposability, and EB property of linear maps.
\end{proposition}

\begin{proof}
It is straightforward from the definitions and closedness of the spaces associated with each of the properties mentioned above. For example, the set of all decomposable linear maps $\Le:B(H_A)\to B(H_B)$ is closed in $B(B(H_A),B(H_B))$ with respect to the natural (Euclidean) topology.
\end{proof}

For a linear map $\Le:B(H_A)\to B(H_B)$, the \textit{adjoint map} $\Le^*:B(H_B)\to B(H_A)$ of $\Le$ is a linear map satisfying
\begin{equation}
    \text{Tr}(\Le(X)\,Y)=\text{Tr}(X\,\Le^*(Y))
\end{equation}
for all $X\in B(H_A)$ and $Y\in B(H_B)$. Recall that the adjoint operation $\Le\mapsto \Le^*$ preserves positivity, CP, CCP, PPT, and decomposability.

\begin{proposition} \label{prop-twirling}
Let $\pi:G\to \mathcal{U}(H)$, $\pi_A:G\to \mathcal{U}(H_A)$ and $\pi_B:G\to \mathcal{U}(H_B)$ be unitary representations of $G$. Then we have the following.
\begin{enumerate}
    \item $\text{Tr}((\mathcal{T}_{\pi}X)\,Y)=\text{Tr}(X (\mathcal{T}_{\pi}Y))$ for any $X,Y\in B(H)$.
    
    \item $\mathcal{T}_{\pi_A\otimes \pi_B}\circ (T_A\otimes \text{id}_B)=(T_A\otimes \text{id}_B)\circ \mathcal{T}_{\overline{\pi_A}\otimes \pi_B}$ where $T_A$ is the transpose on $B(H_A)$.
    
    \item $\left (\mathcal{T}_{\pi_A,\pi_B}\Le\right )^*=\mathcal{T}_{\pi_B,\pi_A}(\Le^*)$ for any linear map $\Le:B(H_A)\to B(H_B)$.
    
    \item The Choi matrix of $\mathcal{T}_{\pi_A,\pi_B}\Le$ is given by $\mathcal{T}_{\overline{\pi_A}\otimes \pi_B}\left ( C_{\Le}\right )$ for any linear map $\Le:B(H_A)\to B(H_B)$.
\end{enumerate}
\end{proposition}
\begin{proof}
(1) Since the Haar measure on the compact group $G$ is invariant under the inverse $x\mapsto x^{-1}$, we have
\begin{align}
   \text{Tr}((\mathcal{T}_{\pi}X)\,Y)&=\int_G \text{Tr}(\pi(x)X\pi(x^{-1})Y)dx\\
    &=\text{Tr}\left(X \int_G \pi(x^{-1})Y \pi(x) dx\right)\\
    &=\text{Tr}\left(X \int_G \pi(x)Y \pi(x^{-1}) dx\right)=\text{Tr}(X (\mathcal{T}_{\pi}Y))
\end{align}
for any $X,Y\in B(H)$.

(2) It suffices to show the equality for product operators $X=P\otimes Q$, and the conclusion follows immediately from the observation 
\begin{equation}
    \pi_A(x)P^T\pi_A(x)^*= \left (\overline{\pi_A(x)}P\pi_A(x)^T\right )^T.
\end{equation}

(3) For any $X\in B(H_A)$ and $Y\in B(H_B)$, we have
\begin{align}
    &\text{Tr}(X \left ( \mathcal{T}_{\pi_B,\pi_A}\Le^*\right )(Y))\\
    &=\int_G \text{ Tr}(X\pi_A(x)^*\Le^*(\pi_B(x)Y\pi_B(x)^*)\pi_A(x))dx\\
    &=\int_G \text{ Tr}(\pi_B(x)^*\Le(\pi_A(x)X\pi_A(x)^*)\pi_B(x)\,Y)dx\\
   & =\text{Tr}(\left ( \mathcal{T}_{\pi_A,\pi_B}\Le\right )(X)\, Y),
\end{align}
which gives us the desired conclusion.

(4) First of all, note that 
\begin{align}
 \label{eq200}&\sum_{i,j=1}^{d_A}(\overline{\pi_A(x)}e_{ij}\pi_A(x)^t)\otimes (\pi_B(x)\Le(e_{ij})\pi_B(x)^*)\\
 \label{eq210}&=\sum_{i,j=1}^{d_A}e_{ij}\otimes (\pi_B(x)\Le(\pi_A(x)^*e_{ij}\pi_A(x))\pi_B(x)^*).
\end{align}
for each $x\in G$. Indeed, the LHS \eqref{eq200} can be understood as
\begin{align}
   \label{eq20} &d_A(\text{id}_A\otimes (\text{Ad}_{\pi_B(x)}\circ \Le))\left((\overline{\pi_A(x)}\otimes \text{Id}_A)|\Omega_A\ra\la\Omega_A|(\pi_A(x)^t\otimes \text{Id}_A)\right),
   \end{align}
   and the RHS \eqref{eq210} can be understood as
   \begin{align}\label{eq21} 
   &d_A(\text{id}_A\otimes (\text{Ad}_{\pi_B(x)}\circ \Le))\left((\text{Id}_A\otimes \pi_A(x)^*)|\Omega_A\ra\la\Omega_A|(\text{Id}_A\otimes \pi_A(x))\right)
\end{align}
where $\text{Ad}_{V}(Y)=VY V^*$. Moreover, the so-called {\it ricochet property}
\begin{equation}
    (X\otimes \text{Id}_A)|\Omega_A\ra=(\text{Id}_A\otimes X^t)|\Omega_A\ra,\;\; X\in B(H_A),
\end{equation}
implies \eqref{eq20} = \eqref{eq21}. Finally, taking the Haar integral on both sides completes the proof.
\end{proof}

Combining Proposition \ref{prop-twirling} (2), (3), and (4) with the fact that both $\text{Inv}(\pi_A\otimes \pi_B)$ and $\text{Cov}(\pi_A,\pi_B)$ are the images of the twirling projections, we obtain  the following useful properties.

\begin{corollary} \label{cor-twirling}
Let $X\in B(H_A\otimes H_B)$ be a bipartite operator and $\Le:B(H_A)\to B(H_B)$ be a linear map. Then
\begin{enumerate}
    \item $X\in \text{Inv}(\pi_A\otimes \pi_B)$ if and only if $(T_A\otimes {\text id})(X)\in \text{Inv}(\overline{\pi_A}\otimes \pi_B)$. 
    
    \item $\Le\in \text{Cov}(\pi_A,\pi_B)$ if and only if $\Le^*\in \text{Cov}(\pi_B,\pi_A)$.
    
    \item $\Le\in \text{Cov}(\pi_A,\pi_B)$ if and only if $C_{\Le}\in \text{Inv}(\overline{\pi_A}\otimes \pi_B)$.
\end{enumerate}
\end{corollary}

%\begin{proposition}
%\begin{enumerate}
%    \item $\Le \in {\rm Cov}(\pi_A,\pi_B)$ if and only if $C_{\Le}\in {\rm Inv}(\overline{\pi_A}\otimes \pi_B)$.
    
%    \item If $\Le\in {\rm Cov}(\pi_A,\pi_B)$, then $\Le^*\in {\rm Cov}(\pi_B,\pi_A)$.
%\end{enumerate}
%\end{proposition} 

\begin{remark}
The results in Corollary \ref{cor-twirling} have been noted in various contexts, \cite[Lemma 6]{EW01}, \cite[Lemma 11]{GBW21}, and \cite[Proposition 5.1, Theorem 3.5]{LeYo22} for examples. Moreover, extendibility to more general contexts of compact quantum group symmetry was proved in \cite{LeYo22}.
%and proves the analogue of Lemma \ref{lem-covIrred}.
\end{remark}

\section{A framework to characterize entanglement under group symmetry} \label{sec-framework}

%For a separable quantum state $\rho=\sum_{i=1}^n p_i \rho^A_i\otimes \rho^B_i\in \mathcal{D}(H_A\otimes H_B)$, it is straightforward to see that 
%\begin{equation}
%(\text{id}_A\otimes \Le)(\rho)=\sum_{i=1}^n p_i \rho^A_i\otimes \Le (\rho^B_i ) \geq 0
%\end{equation}
%for any positive map $\mathcal{L}:B(H_B)\rightarrow B(H_A)$.
%Moreover, the converse is also true 
Let us recall a result of Horodecki on the characterization of entanglement \cite{HHH96}: a bipartite quantum state $\rho\in \mathcal{D}(H_A\otimes H_B)$ is separable if and only if $({\text {id}_A}\otimes \Le)(\rho)\geq 0$ for all positive linear maps $\mathcal{L}: B(H_B)\rightarrow B(H_A)$. Indeed, by duality arguments, the authors showed that they are also equivalent to seemingly a weaker condition `$\la \rho,\mathcal{L}\ra\geq 0$'. Here, the dual pairing $\la \cdot, \cdot\ra$ is defined by
\begin{equation}
\la X,\mathcal{N}\ra=\text{Tr}((\text{id}_A\otimes \mathcal{N})(X)\, |\Omega_A\ra\la \Omega_A|)= \text{Tr}(XC_{\mathcal{N}^*})
\end{equation}
for any operator $X\in B(H_A\otimes H_B)$ and linear map $\mathcal{N}:B(H_B)\rightarrow B(H_A)$. In other words, positive linear maps play a crucial role as detectors for the bipartite entanglement, in the sense that there should exist a positive linear map $\Le$ such that $(\text{id}\otimes \Le)(\rho)$ is non-positive whenever $\rho$ is entangled.

One technical issue in this characterization is that verifying whether a linear map is positive or not is computationally intractable \cite{La10,NiZh16}. 
%{\color{blue} Furthermore, it is shown in \cite{Sko16} that if $d_A,d_B\geq 3$, then there is no family of finitely many detectors $\set{\Le_1,\ldots, \Le_N}$ satisfying that,  for every $\rho\in \D(H_A\otimes H_B)$, $\rho$ is separable if and only if $(\text{id}_A\otimes \Le_i)(\rho)\geq 0$ for $i=1,\ldots, N$.}
However, one of the main purposes of this paper is to develop a universal framework to characterize separability of invariant quantum states. The key idea is that, for $\rho\in \text{InvQS}(\overline{\pi_A}\otimes\pi_B)$, it is enough to consider only $\mathcal{L}\in \text{CovPos}(\pi_B,\pi_A)$ to investigate separability of $\rho$. Let us begin with a simple and useful lemma.

\begin{lemma} \label{lem-paring}
For any bipartite operator $X\in B(H_A\otimes H_B)$ and linear map $\Le:B(H_B)\to B(H_A)$, we have
\begin{equation}
    \la \mathcal{T}_{\overline{\pi_A}\otimes \pi_B}X, \Le \ra=\la X, \mathcal{T}_{\pi_B,\pi_A}\Le\ra.
\end{equation}
\end{lemma}

\begin{proof}
Thanks to Proposition \ref{prop-twirling}, we have
\begin{align}
    \la \mathcal{T}_{\overline{\pi_A}\otimes \pi_B}X, \Le \ra&=\text{Tr}((\mathcal{T}_{\overline{\pi_A}\otimes \pi_B}X)C_{\Le^*})\\ \nonumber
    &=\text{Tr}(X(\mathcal{T}_{\overline{\pi_A}\otimes \pi_B}C_{\Le^*}))=\text{Tr}(X C_{\overline{\Le}^*})=\la X, \overline{\Le}\ra,
\end{align}
where $\overline{\Le}=(\mathcal{T}_{\pi_A,\pi_B}\Le^*)^*=\mathcal{T}_{\pi_B,\pi_A}\Le$.
\end{proof}

Then Lemma \ref{lem-paring} allows us to conclude that covariant positive linear maps are enough to characterize separability of bipartite invariant quantum states. We remark that the ideas of the following proof appeared already for some specified symmetries \cite{Ka11,Gu11,SiNe21}. 

\begin{theorem}\label{cor-EBT}
Let $\pi_A:G\rightarrow B(H_A)$ and $\pi_B:G\rightarrow B(H_B)$ be unitary representations, and let $\rho \in \text{ InvQS}(\overline{\pi_A}\otimes \pi_B)$. The following are equivalent.
\begin{enumerate}
    \item $\rho$ is a separable quantum state.
    
    \item $({\text {id}}_A\otimes \Le)(\rho) \geq 0$ in $B(H_A\otimes H_A)$ for any $(\pi_B,\pi_A)$-covariant positive linear map $\mathcal{L}:B(H_B)\rightarrow B(H_A)$.
    
    \item $\la \rho, \Le \ra \geq 0$ for any $(\pi_B,\pi_A)$-covariant positive linear map $\Le:B(H_B)\to B(H_A)$.
\end{enumerate}
\end{theorem} 

\begin{proof}
Two directions $(1)\Rightarrow (2)$ and $(2)\Rightarrow (3)$ are clear. For the last direction $(3)\Rightarrow (1)$, let us show that $\la \rho, \Le \ra \geq 0$ for all positive linear maps $\Le:B(H_B)\to B(H_A)$. Indeed, since $\rho$ is $\overline{\pi_A}\otimes \pi_B$-invariant and $\mathcal{T}_{\pi_B,\pi_A}\Le\in \text{CovPos}(\pi_B,\pi_A)$, we can apply Lemma \ref{lem-paring} to obtain
\begin{equation}
    \la \rho,\Le \ra=\la \mathcal{T}_{\overline{\pi_A}\otimes \pi_B}\rho ,\Le \ra = \la \rho, \mathcal{T}_{\pi_B,\pi_A}\Le\ra\geq 0.
\end{equation}
\end{proof}

From now on, let us focus on the question of whether PPT property coincides with separability, i.e. problem PPT=SEP for invariant quantum states. Recall that the dual notions of PPT property and separability correspond to decomposability and positivity respectively. Indeed, many duality arguments \cite{Kye23} carry over into our framework, and the problem PPT=SEP in $\text{InvQS}(\overline{\pi_A}\otimes \pi_B)$ is equivalent to the question whether all positive maps are decomposable, i.e. problem POS=DEC in $\text{Cov}(\pi_B,\pi_A)$.

\begin{proposition} \label{prop-duality}
Let $\Le: B(H_B)\to B(H_A)$ be $(\pi_B,\pi_A)$-covariant. Then
\begin{enumerate}
    \item $\Le$ is positive if and only if $(\text{id}_A\otimes \Le)(\rho)\geq 0$ for any separable $\rho\in \text{InvQS}(\overline{\pi_A}\otimes \pi_B)$.
    
    \item $\Le$ is decomposable if and only if $(\text{id}_A\otimes \Le)(\rho)\geq 0$ for any PPT $\rho\in \text{InvQS}(\overline{\pi_A}\otimes \pi_B)$.
\end{enumerate}
\end{proposition}
\begin{proof}
(1) Suppose $(\text{id}_A\otimes \Le)(\rho)\geq 0$ for any separable $\rho\in \text{InvQS}(\pi_A\otimes \pi_B)$. Then for every separable state $\rho\in \D(H_A\otimes H_B)$, we have
\begin{equation}
    \la \rho,\Le \ra=\la \rho, \mathcal{T}_{\pi_B,\pi_A}\Le\ra = \la \mathcal{T}_{\overline{\pi_A}\otimes \pi_B}\rho ,\Le \ra \geq 0
\end{equation}
by Lemma \ref{lem-paring} and by the separability of $\mathcal{T}_{\overline{\pi_A}\otimes \pi_B}\rho$. Now positivity of $\Le$ follows from \cite[Theorem 3.1]{EK00}. The converse direction is clear.

(2) It is enough to repeat the arguments from (1) based on the following duality result \cite{Sto82}: $\Le$ is decomposable if and only if $\la \rho, \Le\ra\geq 0$ for every PPT state $\rho$.
\end{proof}

\begin{corollary} \label{cor-CovPPTES}
The following are equivalent:
\begin{enumerate}
    \item PPT=SEP in $\text{InvQS}(\overline{\pi_A}\otimes \pi_B)$.
    
    \item POS=DEC in $\text{Cov}(\pi_B,\pi_A)$.
\end{enumerate}
\end{corollary}
\begin{proof}
((1) $\Rightarrow$ (2)) If $\Le\in \text{CovPos}(\pi_B,\pi_A)$, then $(\text{id}_A\otimes \Le)(\rho)\geq 0$ for every separable (hence every PPT) state $\rho\in \text{InvQS}(\overline{\pi_A}\otimes \pi_B)$. Thus, $\Le$ is decomposable by Proposition \ref{prop-duality}.

((2) $\Rightarrow$ (1)) If $\rho\in \text{InvQS}(\overline{\pi_A}\otimes \pi_B)$ is a PPT state, then $(\text{id}_A\otimes \Le)(\rho)\geq 0$ for every decomposable (hence every positive) linear map $\Le\in \text{Cov}(\pi_B,\pi_A)$. Thus, $\rho$ is separable by Theorem \ref{cor-EBT}.
\end{proof}

Note that PPT=SEP in $\text{InvQS}(\overline{\pi_A}\otimes \pi_B)$, or equivalently POS=DEC in $\text{Cov}(\pi_B,\pi_A)$, implies that PPT property coincides with the entanglement-breaking property, i.e. PPT=EB in $\text{CovQC}(\pi_A,\pi_B)$. Moreover, we have the following characterization of EB property for $\Phi\in \text{CovQC}(\pi_A,\pi_B)$ by Corollary \ref{cor-twirling} (3). 
\begin{corollary}\label{cor-EBT3}
Let $\Phi\in \text{CovQC}(\pi_A,\pi_B)$. Then the following are equivalent.
\begin{enumerate}
    \item $\Phi$ is entanglement-breaking.
    \item $C_{\Le\circ \Phi}=(\text{id}\otimes \Le)(C_{\Phi})\geq 0$ for any $\Le\in \text{CovPos}(\pi_B,\pi_A)$.
    \item $\mathcal{L}\circ \Phi$ is completely positive for any $\Le\in \text{CovPos}(\pi_B,\pi_A)$.
\end{enumerate} 

\end{corollary}

%\begin{corollary}\label{cor-EBT0}
%Let $\pi_A:G\rightarrow B(H_A)$ and $\pi_B:G\rightarrow B(H_B)$ be unitary representations, and let $\Phi:B(H_A)\rightarrow B(H_B)$ be a $(\pi_A,\pi_B)$-covariant quantum channel. Then the following are equivalent.
%\begin{enumerate}
%    \item $\Phi$ is entanglement-breaking.
%    \item $({\rm id}\otimes \Le)(C_{\Phi})\geq 0$ for any $(\pi_B,\pi_A)$-covariant positive linear map $\mathcal{L}:B(H_B)\rightarrow B(H_A)$.
%\end{enumerate}
%\end{corollary}

To summarize, we have 
\begin{center}
    PPT=SEP in $\text{InvQS}(\overline{\pi_A}\otimes \pi_B)$ $\Leftrightarrow$ DEC=POS in $\text{Cov}(\pi_B,\pi_A)$,
\end{center}
and these conditions imply PPT=EB in $\text{CovQC}(\pi_A,\pi_B)$. One might ask whether all these three problems are equivalent, but one technical issue is that $\text{CovQC}(\pi_A,\pi_B)$ is not identified with $\text{InvQS}(\overline{\pi_A}\otimes \pi_B)$ in general. This leads us to question whether the (reduced) channel-state duality \begin{equation}
\widetilde{C}: \text{CovQC}(\pi_A,\pi_B)\rightarrow \text{InvQS}(\overline{\pi_A}\otimes \pi_B)    
\end{equation}
is bijective. The channel-state duality $\widetilde{C}$ is not surjective in general, but it is known to be the case if $\pi_A$ is irreducible, as already noted in \cite[Lemma 15]{GBW21}. Moreover, we prove that the converse is also true, i.e. the channel-state duality $\widetilde{C}$ is bijective if and only if $\pi_A$ is irreducible. Let us start with the following lemma.

\begin{lemma} \label{lem-covIrred} Let $\pi_A:G\rightarrow \mathcal{U}(H_A)$ and $\pi_B:G\rightarrow \mathcal{U}(H_B)$ be unitary representations of $G$ and let $\Le\in \text{Cov}(\pi_A,\pi_B)$.
\begin{enumerate}
    \item If $\pi_B$ is irreducible, then $\Le(\text{Id}_A)=c\,\text{Id}_B$ for some constant $c$.
    
    \item If $\pi_A$ is irreducible, then there is a constant $c$ such that $\text {Tr}(\Le(X))=c\,\text {Tr}(X)$ for $X\in B(H_A)$.
\end{enumerate}
\end{lemma}

\begin{proof}
\begin{enumerate}
    \item From the irreducibility of $\pi_B$ and the identity
    \begin{equation}
        \pi_B(x)\Le(\text{Id}_A)\pi_B(x)^*=\Le(\pi_A(x)\pi_A(x)^*)=\Le(\text{Id}_A),
    \end{equation}
    we have $\Le(\text{Id}_A)\in \text{Inv}(\pi_B)=\Comp\cdot \text{Id}_B$.
    \item The adjoint map ${\Le}^*$ is $(\pi_B,\pi_A)$-covariant by Corollary \ref{cor-twirling} (2), so ${\Le}^*({\text {Id}}_B)=c\,{\text {Id}_A}$ for some $c$ by (1). In this case, we have
    \begin{equation}
        \text {Tr}({\Le}(X))=\text {Tr}({\Le}(X) \,\text {Id}_B)=\text {Tr}(X\,\Le^*(\text {Id}_B))=c\,\text {Tr}(X)
    \end{equation}
for any $X\in B(H_A)$.
\end{enumerate}
\end{proof}

Now, let us apply Lemma \ref{lem-covIrred} (2) to prove that the channel-state duality $\widetilde{C}:\text{CovQC}(\pi_A,\pi_B)\rightarrow \text{InvQS}(\overline{\pi_A}\otimes \pi_B)$ should be bijective if and only if $\pi_A$ is irreducible.

\begin{proposition}\label{rmk30}
Let $\pi_A:G\rightarrow \mathcal{U}(H_A)$ and $\pi_B:G\rightarrow \mathcal{U}(H_B)$ be unitary representations of $G$ and let $\Le\in \text{Cov}(\pi_B,\pi_A)$. Then the channel-state duality
    \begin{equation}
        \widetilde{C}: \text{CovQC}(\pi_A,\pi_B)\rightarrow \text{InvQS}(\overline{\pi_A}\otimes \pi_B)
    \end{equation}
    is bijective if and only if $\pi_A$ is irreducible.
\end{proposition}

\begin{proof}
Let us prove the if part first. For any $\rho\in \text{InvQS}(\overline{\pi_A}\otimes \pi_B)$ there exists completely positive $\mathcal{L}\in \text{Cov}(\pi_A,\pi_B)$ such that $C_{\mathcal{L}}=\rho$ by Corollary \ref{cor-twirling} (3). Moreover, $\mathcal{L}$ should be trace-preserving. Indeed, irreducibility of $\pi_A$ implies that there exists a constant $c$ such that $\text{Tr}(\Le(X))=c\text{Tr}(X)$ for all $X\in B(H_A)$ by Lemma \ref{lem-covIrred} (2), and we have
\begin{equation}
    c= \frac{c}{d_A} \sum_{i=1}^{d_A}\text{Tr}(e_{ii}) =\frac{1}{d_A}\sum_{i=1}^{d_A}\text{Tr}(e_{ii}\otimes\mathcal{L}(e_{ii}))=\text{Tr}(C_{\mathcal{L}})=1.
\end{equation}
Conversely, if we assume that $\pi_A=\pi_A^{(1)}\oplus \pi_A^{(2)}$ with $H_A=H_A^{(1)}\oplus H_A^{(2)}$ and if $\Pi_1$ is the orthogonal projection from $H_A$ onto $H_A^{(1)}$, then we can take a CP non-TP map $\Le:B(H_A)\to B(H_B)$ given by 
\begin{equation}
    \Le(X)=\frac{d_A}{d_B\cdot \dim H_A^{(1)}}{\rm \text{Tr}}(\Pi_1 X)\text{Id}_B
\end{equation}
whose Choi matrix is 
\begin{equation}
    C_{\Le}=\left (\frac{1}{\dim H_A^{(1)}}\Pi_1\right )\otimes \left (\frac{1}{d_B}\text{Id}_B\right )\in \text{InvQS}(\overline{\pi_A}\otimes \pi_B).
\end{equation}
\end{proof}

\begin{corollary}\label{cor-equiv}
Let $\pi_A:G\rightarrow \mathcal{U}(H_A)$ and $\pi_B:G\rightarrow \mathcal{U}(H_B)$ be unitary representations of $G$ and suppose that $\pi_A$ is irreducible. Then the following are equivalent.
\begin{enumerate}
    \item PPT=SEP in $\text{InvQS}(\overline{\pi_A}\otimes \pi_B)$.
    \item PPT=EB in $\text{CovQC}(\pi_A,\pi_B)$.
    \item POS=DEC in $\text{CovPos}_1(\pi_B,\pi_A)$.
\end{enumerate}
\end{corollary}
\begin{proof}
It is enough to note that Lemma \ref{lem-covIrred} allows us to focus on a smaller convex set $\text{CovQC}(\pi_A,\pi_B)$ and $\text{CovPos}_1(\pi_B,\pi_A)$ rather than $\text{Cov}(\pi_A,\pi_B)$ and $\text{CovPos}(\pi_B,\pi_A)$, respectively, in Corollary \ref{cor-CovPPTES}.
\end{proof}

%Note that $\text{CovPos}(\pi_B,\pi_A)$ plays crucial roles in solving three equivalent problems in Corollary \ref{cor-equiv}. Another important conclusion coming from the irreducibility of $\pi_A$ (resp. $\pi_B$) is that we can further restrict ourselves to a smaller convex set $\text{CovPos}_1(\pi_B,\pi_A)$ (resp. $\text{CovPosTP}(\pi_B,\pi_A)$) thanks to Lemma \ref{lem-covIrred} as in the proof of Corollary \ref{cor-equiv}. 
Moreover, we prove that the extreme points of $\text{CovPos}_1(\pi_B,\pi_A)$ are enough for the entanglement detection, which we propose as a universal machinery to characterize entangled invariant quantum states with general compact group symmetries. Let us denote by $\text{Ext}(S)$ the set of all extreme points of a convex set $S$.

%Theorem \ref{cor-EBT} and Corollary \ref{cor-EBT3} allow us to focus on $\text{CovPos}(\pi_B,\pi_A)$ to analyze separability of $\rho\in \text{InvQS}(\overline{\pi_A}\otimes \pi_B)$ and entanglement-breaking property of $\Phi\in \text{CovQC}(\pi_A,\pi_B)$. Moreover, under a mild assumption of irreducibility of $\pi_A$ (resp. $\pi_B$), we may further restrict ourselves to the convex set $\text{CovPos}_1(\pi_B,\pi_A)$ (resp. $\text{CovPosTP}(\pi_B,\pi_A)$) thanks to the following Lemma \ref{lem-covIrred}. Later, in Theorem \ref{cor-EBT2}, we will demonstrate that their extreme points are enough for the characterization. %({\color{red} remove this sentense?})

\begin{theorem}\label{thm-main}
Let $\pi_A:G\rightarrow B(H_A)$ and $\pi_B:G\rightarrow B(H_B)$ be unitary representations, and suppose that $\pi_A$ is irreducible. Let $\rho\in \text{InvQS}(\overline{\pi_A}\otimes \pi_B)$ and $\Phi\in \text{CovQC}(\pi_A,\pi_B)$ such that $C_{\Phi}=\rho$ 
from Proposition \ref{rmk30}. The following are equivalent.
\begin{enumerate}
    \item $\rho$ is separable.
    \item $(\text{id}\otimes \Le)(\rho)\geq 0$ for any $\Le\in \text{CovPos}_1(\pi_B,\pi_A)$.
    \item $(\text{id}\otimes \Le)(\rho)\geq 0$ for any $\Le\in \text{Ext}\left (\text{CovPos}_1(\pi_B,\pi_A)\right )$.
    \item $\Phi$ is entanglement-breaking.
    \item $\Le\circ \Phi$ is completely positive for any $\Le\in \text{CovPos}_1(\pi_B,\pi_A)$. 
    \item $\Le\circ \Phi$ is completely positive 
    for any $\Le\in \text{Ext}\left (\text{CovPos}_1(\pi_B,\pi_A)\right )$. 
    \end{enumerate}
%Here, $\text{CovPos}_1(\pi_B,\pi_A)$ is a scalar multiple of $\text{CovPosTP}(\pi_B,\pi_A)$ if both $\pi_A$ and $\pi_B$ are irreducible.
\end{theorem}

\begin{proof}
The equivalences (1) $\Leftrightarrow$ (4), (2) $\Leftrightarrow$ (5), (3) $\Leftrightarrow$ (6) and one-side implications (1) $\Rightarrow$ (2) $\Rightarrow$ (3) are clear. Moreover, the direction (2) $\Rightarrow$ (1) follows from Theorem \ref{cor-EBT} and Lemma \ref{lem-covIrred} (1). For the proof of (3) $\Rightarrow$ (2), note that ${\text{CovPos}}_1(\pi_B,\pi_A)$ is a compact subset of
\begin{equation}
    \set{\Le\in B(B(H_B),B(H_A)): \|\Le\|_{\rm op}\leq 1},
\end{equation}
where $\|\cdot \|_{\rm op}$ denotes the norm of linear operators between $C^*$-algebras $B(H_B)$ and $B(H_A)$, since the positivity of $\Le$ implies $\|\Le\|_{\rm op}=\|\Le(\text{Id}_B)\|=1$ \cite[Corollary 2.9]{Pau02}. 
%Moreover,  Since both $\text{CovPos}_1(\pi_B,\pi_A)$ and $\text{CovPosTP}(\pi_B,\pi_A)$ are convex and compact,
Therefore, $\text{CosPos}_1(\pi_B,\pi_A)$ can be written as a convex hull of its extreme points, which completes the proof.
\end{proof}

\begin{remark}
If $\pi_B$ is irreducible instead of irreducibility of $\pi_A$, then $\Phi$ is chosen to be a unital CP map up to constant, and $\text{CovPosTP}(\pi_B,\pi_A)$ replaces the role of $\text{CovPos}_1(\pi_B,\pi_A)$ in (2), (3), (5), (6). Note that compactness of $\text{CovPosTP}(\pi_B,\pi_A)$ comes from the identification with $\text{CovPos}_1(\pi_A,\pi_B)$ (up to constant) via taking the adjoint operation.
\end{remark}

%Recall that the set $\text{CovPos}_1(\pi_B,\pi_A)$ is enough to study POS=DEC problem, and it is enough to consider only its extremal elements.
Finally, we claim that the decomposability of the extremal elements in $\text{CovPos}_1(\pi_B,\pi_A)$ is much easier to check thanks to the following theorem.

%We conclude this section by presenting a theorem which may be quite useful when we find a PPT entanglement detector (i.e., positive non-decomposable map).

\begin{theorem}\label{thm-extremal-decomposable}
Suppose that $\pi_A$ (resp. $\pi_B$) is irreducible, and let $\Le\in \text{Ext}(\text{CovPos}_1(\pi_B,\pi_A))$ (resp. $\Le\in \text{Ext}(\text{CovPosTP}(\pi_B,\pi_A))$. Then $\Le$ is decomposable if and only if $\Le$ is CP or CCP.
\end{theorem}
\begin{proof}
Let us focus only on the case where $\pi_A$ is irreducible since the other case is analogous. If $\Le$ is decomposable, then there exist a CP map $\Le_1$ and a CCP map $\Le_2$ such that $\Le=\Le_1+\Le_2$. By taking the twirling operation $\mathcal{T}_{\pi_B,\pi_A}$, we have $\Le=\Le_1'+\Le_2'$ where $\Le_i'=\mathcal{T}_{\pi_B,\pi_A}(\Le_i)\in \text{Cov}(\pi_B,\pi_A)$. Note that $\Le_1'$ is CP and $\Le_2'$ is CCP, and we can further write $\Le_i'=\lambda_i\Le_i''$ for some $\lambda_i\geq 0$, $\lambda_1+\lambda_2=1$, and $\Le_i''\in \text{CovPos}_1(\pi_B,\pi_A)$ by Lemma \ref{lem-covIrred} (1). Then extremality of $\Le$ allows us to conclude that $\Le=\Le_1''$ or $\Le=\Le_2''$, which proves the assertion. The other direction is immediate.
\end{proof}

To summarize, our strategy to study the problems PPT=SEP and PPT=EB consists of the following three independent steps, assuming $\pi_A$ is irreducible.

\begin{itemize}
  \item[{[\bf{Step 1}]}] The first step is to characterize all elements in $\text{CovPos}_1(\pi_B,\pi_A)$ for given specific unitary representations $\pi_A$ and $\pi_B$. {If we take the adjoint operation, this step is equivalent to characterize all elements in $\text{CovPosTP}(\pi_A,\pi_B)$.}
  
  \item[{[\bf{Step 2}]}] The next step is to solve the problem POS=DEC in $\text{CovPos}_1(\pi_B,\pi_A)$ {(or equivalently, the problem POS=DEC in $\text{CovPosTP}(\pi_A,\pi_B)$)}. In particular, for a given extremal element $\Le\in \text{Ext}(\text{CovPos}_1(\pi_B,\pi_A))$, $\Le$ is decomposable if and only if $\Le$ is CP or CCP. If POS=DEC holds, then the problem PPT=SEP in $\text{InvQS}(\overline{\pi_A}\otimes \pi_B)$ and the problem PPT=EB in $\text{CovQC}(\pi_A,\pi_B)$ has the affirmative answer.
  
  \item[{[\bf{Step 3}]}] If there exists a non-decomposable element $\mathcal{L}$ in $\text{CovPos}_1(\pi_B,\pi_A)$, then the last step is to {realize $\mathcal{L}$ as an entanglement detector to find the following PPT entangled objects}:
  \begin{itemize}
      \item $\Phi\in \text{CovPPTQC}(\pi_A,\pi_B)$ for which $\mathcal{L}\circ \Phi$ is non-CP,
      \item $\rho\in \text{InvPPTQS}(\overline{\pi_A}\otimes \pi_B)$ for which $(\text{id}\otimes \mathcal{L})(\rho)\ngeq 0$.
  \end{itemize}
\end{itemize}

\section{PPT=EB holds for $(H,H)$-covariant quantum channels}\label{sec-examples1}

One of the main applications of the results in Section \ref{sec-framework} is a complete characterization of EB property for quantum channels $\Phi:M_d(\Comp)\to M_d(\Comp)$ of the form
\begin{equation}\label{eq-HHcov}
    \Phi(X)=a \frac{\text{Tr}(X)}{d}\text{Id}_d + b X+ c X^T+(1-a-b-c)\,\text{diag}(X).
\end{equation}

The main result of this section is as follows.

\begin{theorem} \label{thm-HHCov}
Let $\Phi$ be a quantum channel of the form \eqref{eq-HHcov}. Then $\Phi$ is entanglement-breaking if and only if $\Phi$ is PPT.
\end{theorem}

\begin{remark} \label{rmk-VW01}

{Note that the quantum channels of the form \eqref{eq-HHcov} under the condition $a+b+c=1$ are called the {\it generalized Werner-Holevo channels}, and their Choi matrices are given by
\begin{equation}\label{eq40}
C_{\Phi}=\frac{1-b-c}{d^2}\text{Id}_d\otimes \text{Id}_d+b|\Om_d\ra\la \Om_d|+\frac{c}{d}F_d,
\end{equation}
where $F_d=\sum_{i,j=1}^d e_{ij}\otimes e_{ji}$ is the flip matrix.
The subclasses corresponding to the cases $b=0$ or $c=0$ are called {\it the Werner states} and {the isotropic states} respectively, and their separability was studied in \cite{We89,HoHo99,wat2018, SiNe21}. Furthermore, it was proved in \cite{VW01} that PPT=SEP holds for all quantum states of the form \eqref{eq40}.
}
\end{remark}

%Although every channel in \eqref{eq-HHcov} is DOC, such covariance is not strong enough to apply Theorem \ref{cor-EBT2} since the corresponding representation is not irreducible. 

A starting point for a proof of Theorem \ref{thm-HHCov} is to observe that any quantum channel of the form \eqref{eq-HHcov} is covariant with respect to the \textit{signed symmetric group} $\mathcal{H}_d$. One of the equivalent ways to realize the signed symmetric group is to define $\mathcal{H}_d$ as a subgroup of the orthogonal group $\mathcal{O}_d$ generated by permutation matrices and diagonal orthogonal matrices. In other words, every element in $\mathcal{H}_d$ is written as an orthogonal matrix $\displaystyle \sum_{i=1}^d s_i|\sigma(i)\ra\la i|$ for $s_1,s_2,\ldots, s_n\in \left\{\pm 1\right\}$ and $\sigma\in \mathcal{S}_d$. We define $\text{Inv}(H\otimes H)$ and $\text{Cov}(H,H)$ with respect to the fundamental representation $H\in \mathcal{H}_d\mapsto H\in \mathcal{O}_d$, which is irreducible as proved below.
%$\mathcal{S}_d$ be the symmetric group of degree $d$, and let $I_{\sigma}\in \mathcal{O}_d$ be the permutation matrix associated to $\sigma\in \mathcal{S}_d$, determined by $I_{\sigma}|i\ra:=|\sigma(i)\ra$ (or equivalently, $I_{\sigma}=\sum_{i=1}^d |\sigma(i)\ra \la i|$). Moreover, denote $D_x:=\sum_{i=1}^d x_i |i\ra\la i| \in \mathcal{O}_d$ by the diagonal orthogonal matrix associated to the vector $x=(x_1,\ldots, x_d)\in \set{\pm 1}^d$. Note that every diagonal matrix which is also orthogonal is of this form.
%In this section, we consider the class of quantum channels which is covariant under the fundamental representation $\tau^{(d)}: O\in \Hi_d\mapsto O \in \mathcal{O}_d$, where $\Hi_d$ is a subgroup of $\mathcal{O}_d$ generated by permutation matrices $I_{\sigma}$ for $\sigma\in \mathcal{S}_d$ and diagonal orthogonal matrices $D_x$ for $x\in \set{\pm 1}^d$. We call $\Hi_d$ as the \textit{signed symmetric group} of degree $d$ (\textbf{REF?}). Every element of $\Hi_d$ can be described as of the form
%\begin{equation}
%    O=\sum_{i=1}^d x_i |\sigma(i)\ra \la i|.
%\end{equation}
%which justifies the name `signed' symmetric group?? Also 

\begin{lemma}\label{lem20}
The fundamental representation $H\in \mathcal{H}_d\mapsto H\in \mathcal{O}_d$ is irreducible.
\end{lemma}
\begin{proof}
The identity
\begin{equation}
    HXH^T=\sum_{i,j=1}^d s_is_j X_{ij}|\sigma(i)\ra \la \sigma(j)|=\sum_{i,j=1}^d s_{\sigma^{-1}(i)}s_{\sigma^{-1}(j)}X_{\sigma^{-1}(i)\sigma^{-1}(j)}|i\ra\la j|
\end{equation}
and the invariance property $HXH^T=X$ for all $H\in \Hi_d$ tell us that 
\begin{equation}
s_{\sigma(i)}s_{\sigma(j)}X_{\sigma(i)\sigma(j)}=X_{ij}
\end{equation}
for all $s_1,\ldots,s_d\in \set{\pm 1}$ and $\sigma\in \mathcal{S}_d$. This implies that $X_{ii}\equiv X_{11}$ for all $1\leq i\leq d$ and $X_{ij}=0$ for all $i\neq j$, i.e., $X=X_{11}\,\text{Id}_d\in \Comp \cdot \text{Id}_d$.
%, so the fundamental representation $H\in \mathcal{H}_d\mapsto H\in \mathcal{O}_d$ is irreducible.
\end{proof}

Let us denote by $\mathcal{W}$ the space of linear maps spanned by the following four unital TP maps $\psi_0,\psi_1,\psi_2,\psi_3:M_d(\Comp)\to M_d(\Comp)$, where
    \begin{equation} \label{eq-HHCovbasis}
    \left \{\begin{array}{llll}
    \psi_0(X)= \frac{\text{Tr}(X)}{d} \text{Id}_d ,\\
    \psi_1(X)=X,\\
    \psi_2(X)=X^T,\\
    \psi_3(X)= \text{diag}(X)=\sum_{i=1}^d X_{ii}|i\ra\la i|.
    \end{array} \right.
\end{equation}
It is straightforward to check $\psi_i\in \text{Cov}(H,H)$ for $i=0,\ldots, 3$, so we have $\mathcal{W}\subseteq \text{Cov}(H,H)$. To prove $\text{Cov}(H,H) =\mathcal{W}$, let us note the fact that any $\Le\in\text{Cov}(H,H)$ satisfies the so-called \textit{diagonal orthogonal covariance (DOC)} property, i.e.
\begin{equation}
\Le(ZXZ^T)=Z\Le(X)Z^T    
\end{equation}
for all $X\in M_d(\Comp)$ and diagonal orthogonal matrices $Z$. This class of channels has been analyzed recently in \cite{SiNe21, SiNe22b, SDN22}. In particular, it is shown that any DOC map $\mathcal{L}$ can be parameterized by a triple $(A,B,C)\in M_d(\Comp)^3$ satisfying $\text{diag}(A)=\text{diag}(B)=\text{diag}(C)$ such that
\begin{equation} \label{eq-DOC}
    \Le(X)=\text{diag}(A|\text{diag}\, X\ra)+\widetilde{B}\odot X +\widetilde{C}\odot X^T,
\end{equation}
where $|\text{diag}\,Y\ra=\sum_{i=1}^d Y_{ii}|i\ra$, $\widetilde{Y}=Y-\text{diag}(Y)$, and $\odot$ denotes the Schur product (or Hadamard product) between matrices. In this case, let us denote by $\Le=\Le_{A,B,C}$.
%and by $\mathcal{Z}$ the space of all DOC linear maps. Then we have 
%    \begin{equation}
        %\mathcal{W}\subseteq \text{Cov}(H,H)\subseteq \mathcal{Z},
%    \end{equation} 
%    and the associated triples $(A,B,C)$ from the subspace $\mathcal{W}$ are characterized as follows.
%\begin{lemma}
%Let $\mathcal{L}=\mathcal{L}_{A,B,C}\in \mathcal{Z}$. Then $\mathcal{L}\in \mathcal{W}$ if and only if $A,B,C\in \text{span}\left\{J_d,\text{Id}_d\right\}$, where $J_d=\sum_{i,j=1}^d e_{ij}$. 
%In particular, $\mathcal{L}_{A,B,C}$ is trace-preserving if and only if     \begin{equation} \label{eq-HHCovDOC}
%    (A,\widetilde{B},\widetilde{C})=\left(\frac{a}{d}J_d+(1-a)\text{Id}_d, b(J_d-\text{Id}_d),c(J_d-\text{Id}_d)\right)
%\end{equation}
%for some $a,b,c\in \Comp$.
%\end{lemma}
%\begin{proof}
%    ({\bf Need to add some explanations})
%\end{proof}

%    . Conversely, if
 %   \begin{equation}
  %      (A,B,C)=\left(\frac{a}{d}J_d+a'\text{Id}_d,bJ_d+b'\text{Id}_d, cJ_d+c'\text{Id}_d\right),
%    \end{equation}
  %      then $(a',b',c')$ is completely determined by $(a,b,c)$ under the TP condition of $\Le_{A,B,C}$ (recall that $\text{diag(A)=diag(B)=diag(C)}$).
  %  Consequently, our class of channels is precisely the subclass of DOC channels where $A,B,C\in \text{span}\set{J_d,\text{Id}_d}$.

%Now we are ready to show $\mathcal{W}=\text{Cov}(H,H)$ based on Lemma \ref{lem20}.
 
\begin{proposition}\label{prop41}
The space $\text{Cov}(H,H)$ is spanned by the four unital TP positive maps $\psi_0,\psi_1,\psi_2$, and $\psi_3$ from \eqref{eq-HHCovbasis}.
\end{proposition}

\begin{proof}
We already know $\mathcal{W}\subseteq \text{Cov}(H,H)$. To show the reverse inclusion, let us pick an arbitrary $\Le\in \text{Cov}(H,H)$. Since $\Le$ is DOC, there exists $(A,B,C)\in M_d(\Comp)^{3}$ such that 
%$\text{diag}(A)=\text{diag}(B)=\text{diag}(C)$ and
$\mathcal{L}=\mathcal{L}_{A,B,C}$ of the form \eqref{eq-DOC}. Note that $\mathcal{L}$ further satisfies
\begin{equation} \label{eq-tauCov}
    \Le(P_{\sigma}X P_{\sigma}^T)=P_{\sigma}\Le(X)P_{\sigma}^T
\end{equation}
for all $X\in M_d(\Comp)$ and $\sigma\in \mathcal{S}_d$. Here, $P_{\sigma}=\sum_{i=1}^d|\sigma(i)\ra\la i|$ is the permutation matrix associated with $\sigma$. 
%Now, it suffices to prove $A,B,C\in \text{span}\left\{J_d,\text{Id}_d\right\}$ in view of Lemma \ref{lem20}.

Let us take $X=e_{ij}$. If $i=j$, then  \eqref{eq-tauCov} implies 
\begin{equation}
\sum_{k=1}^d A_{k\sigma(i)}|k\ra \la k|=\sum_{k=1}^d A_{ki}|\sigma(k)\ra \la \sigma(k)|,    
\end{equation} 
which means that $A_{ik}= A_{\sigma(i)\sigma(k)}$ for all $1\leq i,k\leq d$ and $\sigma\in \mathcal{S}_d$. Therefore, $A_{ii}\equiv A_{11}$ for all $i$ and $A_{ik}\equiv A_{12}$ for all $i\neq k$.
%, so $A\in \text{span}\left\{J_d,\text{Id}_d\right\}$.
On the other hand, if $i\neq j$, then \eqref{eq-tauCov} becomes 
\begin{align}
    &B_{\sigma(i)\sigma(j)}|\sigma(i)\ra \la \sigma(j)|+C_{\sigma(j)\sigma(i)}|\sigma(j)\ra \la \sigma(i)|\\
    &=B_{ij}|\sigma(i)\ra \la \sigma(j)|+C_{ji}|\sigma(j)\ra \la \sigma(i)|,
\end{align}
which gives $B_{ij}\equiv B_{12}$ and $C_{ij}\equiv C_{12}$ for all $i\neq j$. Consequently, the formula \eqref{eq-DOC} now gives
\begin{equation}
\Le=dA_{12}\psi_0+B_{12}\psi_1+C_{12}\psi_2+(A_{11}-A_{12}-B_{12}-C_{12})\psi_3\in \mathcal{W},
\end{equation}
which in turn shows $\text{Cov}(H,H)\subseteq \mathcal{W}$.
%Moreover, since $\text{diag}(B)=\text{diag}(C)=\text{diag}(A)=A_{11}\text{Id}_d$, we can conclude that $A,B,C\in \text{span}\set{J_d, \text{Id}_d}$.
\end{proof}

%It was proved in \cite{LeYo22} that the four extreme elements in $\text{CovQC}(\tau^{(d)},\tau^{(d)})$ are given by

From now, let us denote $(H,H)$-covariant unital (and TP) maps by
\begin{equation}
    \psi_{a,b,c}=a\psi_0+b\psi_1+c\psi_2+(1-a-b-c)\psi_3
\end{equation}
for simplicity, where $\psi_0,\ldots, \psi_3$ are from \eqref{eq-HHCovbasis}. %We first look at CPTP and PPT condition of $\Phi=\psi_{a,b,c}$. 
By recalling that $\psi_{a,b,c}$ can be understood as a DOC map $\mathcal{L}_{A,B,C}$ {and that complete positivity of DOC maps is fully characterized in \cite[Section 6]{SiNe21}, we can show that $\psi_{a,b,c}$ is CPTP if and only if \begin{equation}\label{eq-HHCP}
\left \{\begin{array}{ccc} 
\quad\; 0 \leq a\leq \frac{d}{d-1},\\
\;\;\;\frac{a}{d}-\frac{1}{d-1}\leq b\leq 1-\frac{d-1}{d}a, \\
-\frac{a}{d}\leq c\leq \frac{a}{d}.
\end{array} \right. 
\end{equation}
}
Note that the set of $(a,b,c)\in \Real^3$ satisfying \eqref{eq-HHCP} is a tetrahedron depicted in Figure \ref{fig:PPTQC(H,H)}. 
\begin{figure}[htb!]
    \centering
    \includegraphics[scale=0.22]{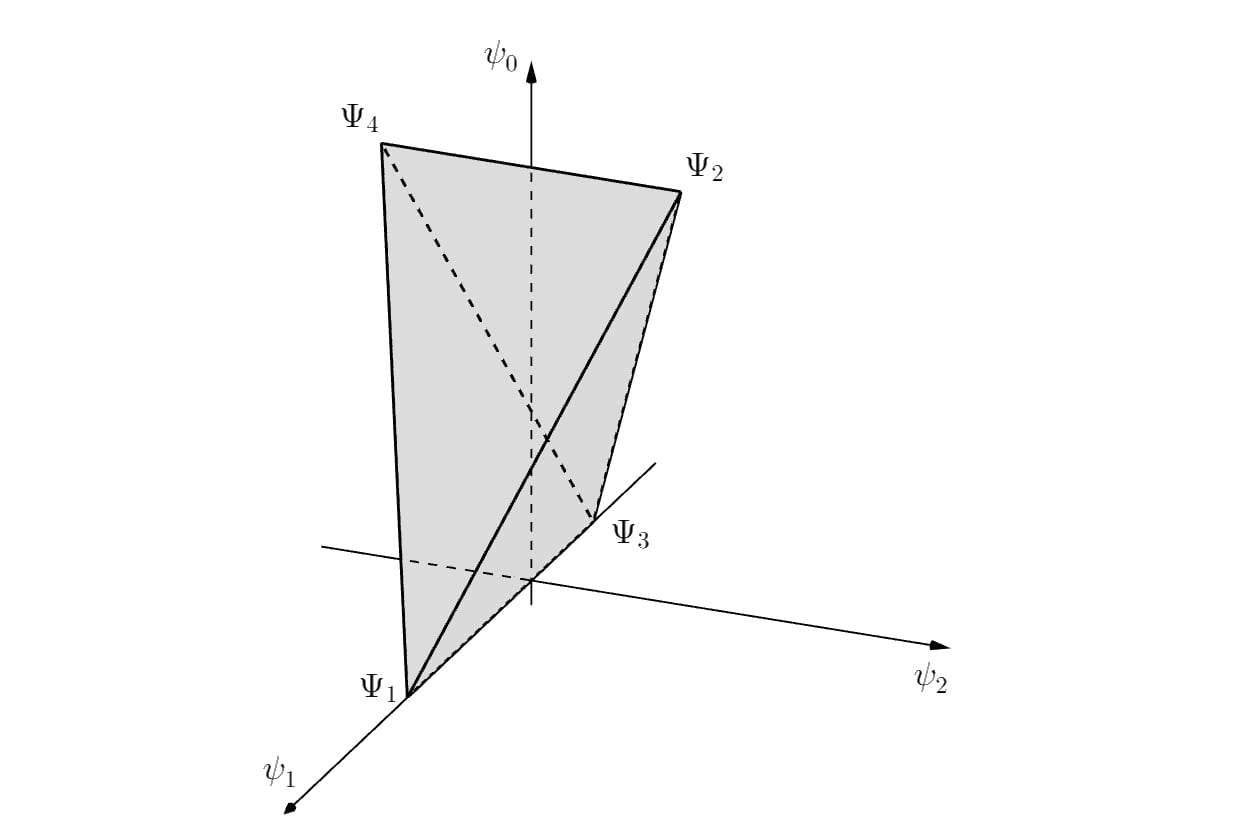}
    \caption{The region of $\text{CovQC}(H,H)$}
    \label{fig:PPTQC(H,H)}
\end{figure}

In particular, there are exactly four extremal $(H,H)$-covariant quantum channels corresponding to the four vertices given by
\begin{equation} \label{eq-HHCPext}
    \begin{cases}
    \Psi_{1}=\psi_{1},\\
    \Psi_{2}=\frac{d}{d-1}\psi_{0}+\frac{1}{d-1}\psi_{2}-\frac{2}{d-1}\psi_{3},\\
    \Psi_{3}=-\frac{1}{d-1}\psi_{1}+\frac{d}{d-1}\psi_{3},\\
    \Psi_{4}=\frac{d}{d-1}\psi_{0} -\frac{1}{d-1}\psi_{1},
    \end{cases}
\end{equation}
whose Choi matrices are (up to normalization) four mutually orthogonal projections. On the other hand, it is easy to see that 
\begin{equation} \label{eq-HHCovTrans}
    T_d\circ \psi_{a,b,c}=\psi_{a,b,c}\circ T_d=\psi_{a,c,b},\;\; a,b,c\in \Comp.
\end{equation}
%{\color{blue}meaning that $T_d(\text{Cov}(H,H))=\text{Cov}(H,H)$ and $\psi_{a,b,c}$ is CCP if and only if $\psi_{a,c,b}$ is CP.}
Therefore, the set of all PPT quantum channels $\psi_{a,b,c}$ is given by
\begin{align}
&\text{CovPPTQC}(H,H)=\text{CovQC}(H,H)\cap T_d\left ( \text{CovQC}(H,H) \right )\nonumber \\
&\;\;=\set{\psi_{a,b,c}: \begin{array}{cc} \quad\quad 0\leq a\leq \frac{d}{d-1},\\\max(\frac{a}{d}-\frac{1}{d-1}, -\frac{a}{d})\leq b,c \leq \min(1-\frac{d-1}{d}a, \frac{a}{d})\end{array}}.
\end{align}
The convex set $\text{CovPPTQC}(H,H)$ can be geometrically understood as the intersection of two tetrahedrons describing the region of CP and CCP $(H,H)$-covariant TP maps (depicted by blue- and red-dotted lines, respectively, in Figure 2). Moreover, if $d\geq 3$, this set has exactly eight vertices (denoted by $v_0,\ldots v_7$).

% \begin{itemize}
%     \item $v_0=(0,0,0)$,
%     \item $v_1=\left (\frac{d}{2(d-1)},\frac{1}{2(d-1)},-\frac{1}{2(d-1)}\right)$, $v_2=\left (\frac{d}{2(d-1)},-\frac{1}{2(d-1)}, \frac{1}{2(d-1)}\right )$,\\
%     $v_3=\left (\frac{d}{2(d-1)},-\frac{1}{2(d-1)},-\frac{1}{2(d-1)}\right )$,
%     \item $v_4=\left (1,\frac{1}{d},\frac{1}{d}\right )$, $v_5=\left (1,\frac{1}{d},-\frac{1}{d(d-1)}\right )$, $v_6=\left (1,-\frac{1}{d(d-1)},\frac{1}{d}\right )$,
%     \item $v_7=\left (\frac{d}{d-1},0,0\right)$.
% \end{itemize}

\begin{figure}[htb!]
    \centering
    \includegraphics[scale=0.22]{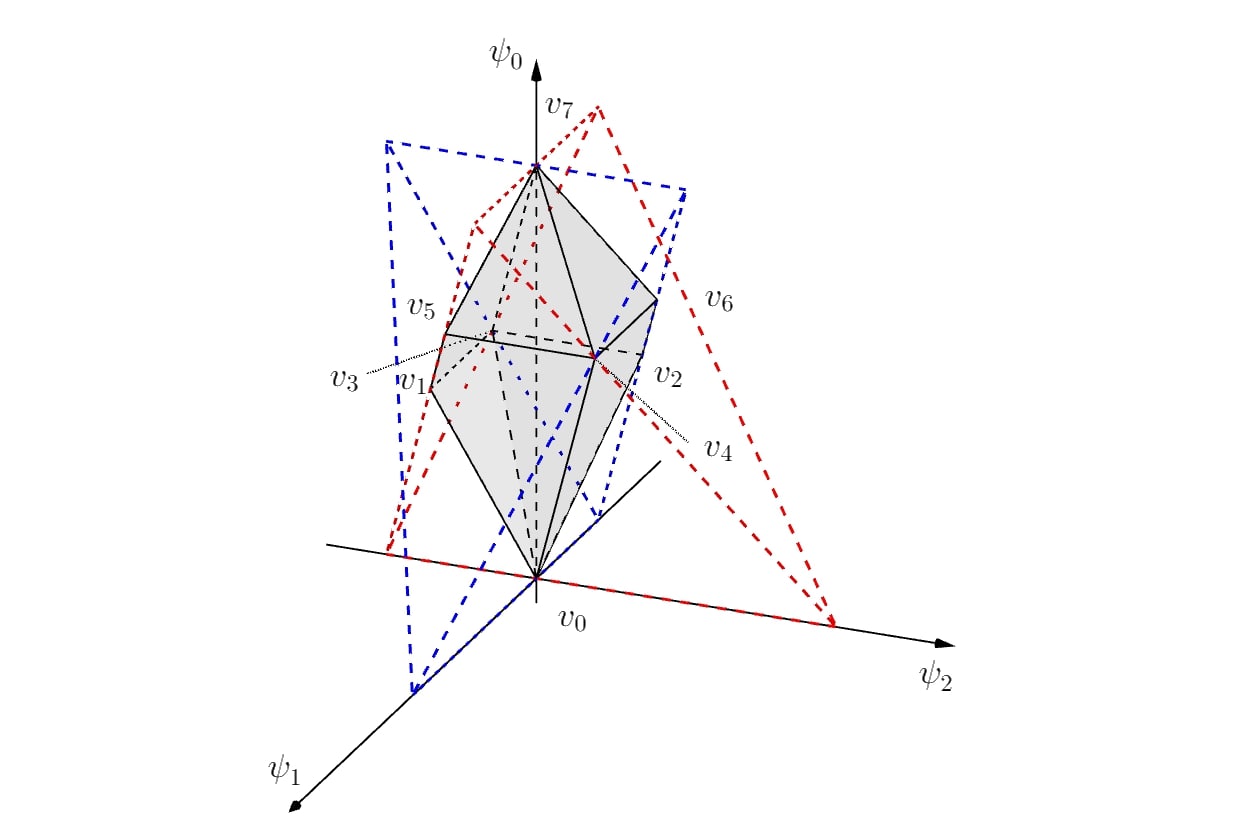}
    \caption{The region of $\text{CovPPTQC}(H,H)$}
    \label{fig:S(d+1)PPTQC}
\end{figure}

We will now explain why the above polytope is precisely identical to the set of all entanglement-breaking $(H,H)$-covariant channels to prove Theorem \ref{thm-HHCov}.

{[\bf{Step 1}+\bf{Step 2}]} {We first characterize the set $\text{CovPos}_1(H,H)$ in terms of the parameters $a,b$, and $c$. Our strategy is to start with the convex hull $\mathcal{V}_d$ of $\text{CovQC}(H,H)\cup T_d\left ( \text{CovQC}(H,H) \right )$, which is an octahedron with eight vertices as exhibited in Figure \ref{fig:S(d+1)Pos}. Then $\mathcal{V}_d\subseteq \text{CovPos}_1(H,H)$ is immediate since any element of $\mathcal{V}_d$ is decomposable. The following Theorem \ref{prop-HHpos} states that these two convex sets coincide, i.e., $\mathcal{V}_d=\text{CovPos}_1(H,H)$.}

\begin{figure}[htb!]
    \centering
    \includegraphics[scale=0.20]{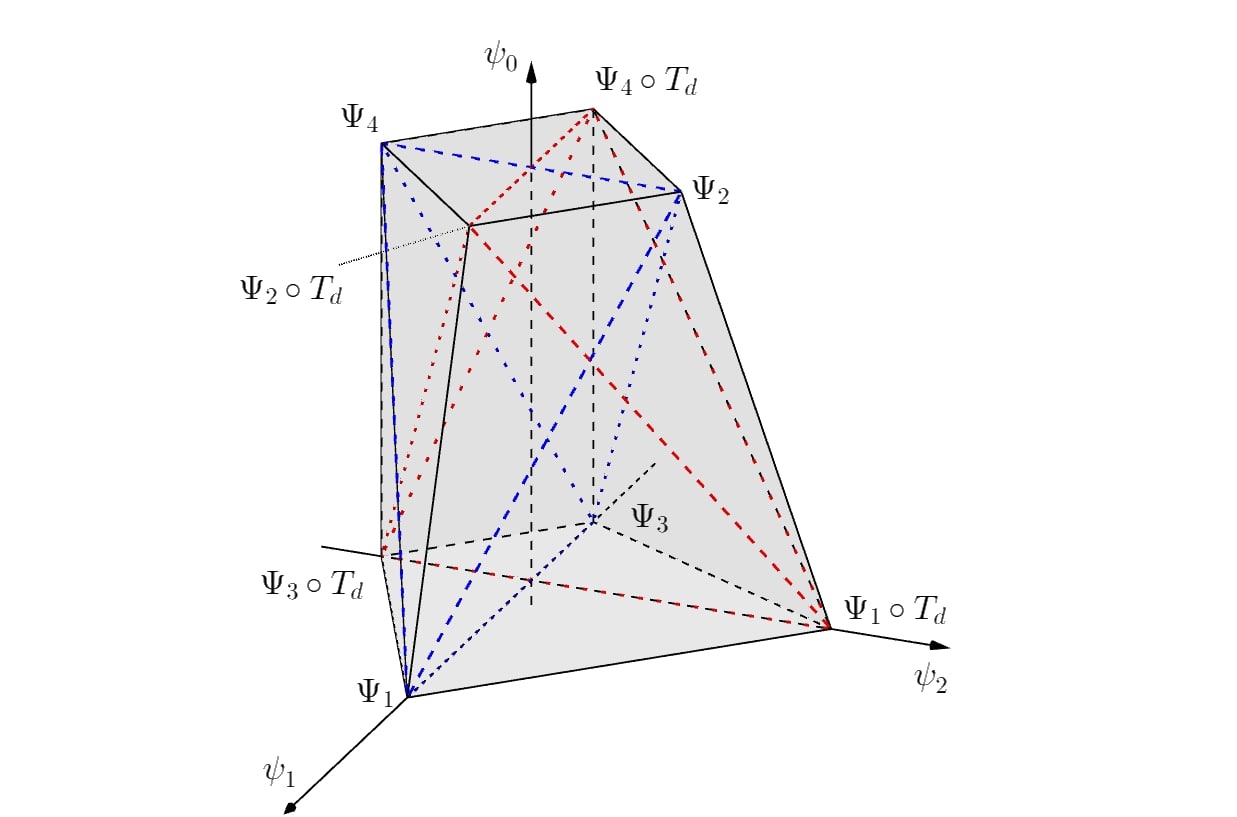}
    \caption{The region of $\text{CovPos}_1(H, H)$}
    \label{fig:S(d+1)Pos}
\end{figure}

\begin{theorem}\label{prop-HHpos}
Let $d\geq 3$. Then the convex set $\text{CovPos}_1(H,H)$ has exactly $8$ extreme points
\begin{equation}
\Psi_1,\; \Psi_2,\;\Psi_3,\;\Psi_4, \quad \Psi_1\circ T_d,\; \Psi_2\circ T_d,\; \Psi_3\circ T_d,\; \Psi_4\circ T_d,
\end{equation}
where $\Psi_1,\ldots, \Psi_4$ are given by \eqref{eq-HHCPext}. In particular, all positive $(H,H)$-covariant maps are decomposable.
\end{theorem}

\begin{proof}
Since $\Psi_1,\ldots, \Psi_4$ are CP and $\Psi_1\circ T_d, \ldots, \Psi_4\circ T_d$ are CCP, the convex hull $\mathcal{V}_d$ of these $8$ maps is obviously contained in $\text{CovPos}_1(H,H)$. To show the reverse inclusion $\text{CovPos}_1(H,H)\subseteq \mathcal{V}_d$, we observe that the set
\begin{equation}
    V_d:=\set{(a,b,c)\in \Real^3\,:\, \psi_{a,b,c}\in \mathcal{V}_d}\subset \Real^3
\end{equation}
is the convex hull of $8$ points 
\begin{equation} \label{eq-HHpos2}
\left \{\begin{array}{lll}
\left(0,1,0\right), \left(\frac{d}{d-1},0, \frac{1}{d-1} \right), \left(0,-\frac{1}{d-1}, 0\right), \left(\frac{d}{d-1}, 0, -\frac{1}{d-1}\right),\\
\left(0,0,1\right), \left(\frac{d}{d-1}, \frac{1}{d-1}, 0\right),  \left(0,0, -\frac{1}{d-1}\right), \left(\frac{d}{d-1}, -\frac{1}{d-1}, 0\right),
\end{array} \right.    
\end{equation}
which are got from \eqref{eq-HHCPext} and \eqref{eq-HHCovTrans}. Therefore, $V_d$ can be understood as the region of $(a,b,c)\in \Real^3$ satisfying the following inequalities:
\begin{equation}\label{eq-HHpos1}
\left \{\begin{array}{lllllll}
(1)& 0 \le a\le \frac{d}{d-1},\\
(2)&\frac{d-2}{d}a+b+c \le 1, \\
(3)& \frac{d-2}{d}a+\left|b-c\right| \le 1, \\
(4)& b+c \ge - \frac{1}{d-1}, \\
(5)& b-\left(d-1\right)c \le 1, \\
(6)& c-\left(d-1\right)b\le1.
\end{array}\right.
\end{equation}
Now if $\psi_{a,b,c}\notin \mathcal{V}_d$ (which is equivalent to $(a,b,c)\notin V_d$, and hence violates at least one of the inequalities (1) - (6) in \eqref{eq-HHpos1}), we can choose a unit vector $\xi\in \Comp^d$ such that $\psi_{a,b,c}(|\xi\ra\la \xi|)$ is not positive semidefinite as in Table \ref{tab:table00}. This shows $\text{CovPos}_1(H,H)\subseteq \mathcal{V}_d$.

\small

\begin{table}[h!]
  \begin{center}
    \caption{Non-positivity outside $V_d$}
    \label{tab:table00}
    \begin{tabular}{|cl|} % <-- Alignments: 1st column left, 2nd middle and 3rd right, with vertical lines in between
       %$\backslash$ 
\hline
$(a,b,c)$ violates (1)&  $|\xi\ra=|1\ra$ $~\Longrightarrow~$ $\psi_{a,b,c}(|\xi\ra\la \xi|)\ngeq 0$\\
\hline
$(a,b,c)$ violates (2)&  $|\xi\ra=\displaystyle \frac{1}{\sqrt{2}}(|1\ra+|2\ra)$ $~\Longrightarrow~$ $\psi_{a,b,c}(|\xi\ra\la \xi|)\ngeq 0$\\
\hline
$(a,b,c)$ violates (3)&  $|\xi\ra=\displaystyle \frac{1}{\sqrt{2}}(|1\ra+ i |2\ra)$ $~\Longrightarrow~$ $\psi_{a,b,c}(|\xi\ra\la \xi|)\ngeq 0$\\
\hline
$(a,b,c)$ violates (4)&  $|\xi\ra =\displaystyle \frac{1}{\sqrt{d}}\sum_{k=1}^{d} |k\ra$ $~\Longrightarrow~$ $\psi_{a,b,c}(|\xi\ra\la \xi|)\ngeq 0$\\
\hline
$(a,b,c)$ violates (5) or (6)& $|\xi\ra=\displaystyle \frac{1}{\sqrt{d}}\sum_{k=1}^{d} e^{\frac{2\pi i k}{d}}|k\ra$ $~\Longrightarrow~$ $\psi_{a,b,c}(|\xi\ra\la \xi|)\ngeq 0$\\
\hline
    \end{tabular}
  \end{center}
\end{table}
\normalsize

\end{proof}

\begin{proof}[\textbf{Proof of Theorem \ref{thm-HHCov}}]
The conclusion is straightforward from Proposition \ref{prop41}, Theorem \ref{prop-HHpos}, and Corollary \ref{cor-equiv}.
\end{proof}

\begin{remark} \label{rmk-KMS20}
\begin{enumerate}
\item Note that Theorem \ref{prop-HHpos} gives a complete characterization of all positive linear maps $\psi$ spanned by $\psi_0,\psi_1,\psi_2,\psi_3$. This strengthens the results in Section 5 of \cite{KMS20} focusing on positive linear maps spanned only by $\psi_0,\psi_1,\psi_3$ without $\psi_2$.
\item Theorem \ref{prop-HHpos} tells us not only POS=DEC, but also explicit decompositions of our positive covariant maps into sums of CP and CCP maps. Note that this was one of the open questions raised in Section 6.c of \cite{KMS20}. We refer to Appendix \ref{sec-CovMU} for more details. 
\end{enumerate}
\end{remark}

\section{The problem PPT=SEP in tripartite systems with unitary group symmetries} \label{sec-examples2}

Recall that a tripartite quantum state $\rho\in \mathcal{D}(H_A\otimes H_B\otimes H_C)$ is called {\it A-BC separable} (resp. {\it A-BC PPT}) if $\rho$ is separable (resp. PPT) in the situation where $B(H_A\otimes H_B\otimes H_C)$ is understood as the bipartite system $B(H_A)\otimes B(H_B\otimes H_C)$. Furthermore, C-AB or B-AC separability (resp. PPT) is defined similarly. We will focus on the situation where $H_A=H_B=H_C=\Comp^d$, and let us denote by 

\begin{equation}
    \left \{ \begin{array}{lll}X^{T_A}=(T_d\otimes \text{id}_{d^2})(X), \\
    X^{T_B}=(\text{id}_d\otimes T_d\otimes \text{id}_d)(X), \\
    X^{T_C}=(\text{id}_{d^2}\otimes T_d)(X), \end{array} \right .
\end{equation} the three partial transposes of $X \in B(H_A\otimes H_B\otimes H_C)=M_{d^3}(\Comp)$.

The main purpose of this section is to apply our results in Section \ref{sec-framework} as new sources to study the problems PPT=SEP, equivalently the problems POS=DEC for some tripartite invariant quantum states. In Section \ref{sec-Tri-Werner}, we exhibit positive non-decomposable covariant maps $\Le:M_d(\Comp)\rightarrow M_{d^2}(\Comp)$ satisfying
\begin{equation}\label{eq53}
\Le(\overline{U}X U^T)=(U\otimes U)\Le(X)(U\otimes U)^*
\end{equation}
for all unitary matrices $U\in \mathcal{U}_d$ and $X\in M_d(\Comp)$. This result is parallel to the fact PPT$\neq$SEP for {\it tripartite Werner states} \cite{EW01}, i.e. tripartite quantum states $\rho\in M_{d^3}(\Comp)$ satisfying
\begin{equation}
(U\otimes U\otimes U)\rho = \rho (U\otimes U\otimes U)
\end{equation}
for all unitary matrices $U\in \mathcal{U}(d)$. 

On the other hand, in Section \ref{sec:Q.orthogonal}, we show that a strong contrast PPT=SEP holds for {\it quantum orthogonally invariant} quantum states. More generally, we prove that PPT=SEP holds for any tripartite quantum states $\rho\in M_{d^3}(\Comp)$ satisfying
\begin{equation}
(U\otimes \overline{U}\otimes U)\rho = \rho (U\otimes \overline{U}\otimes U)
\end{equation}
for all unitary matrices $U\in \mathcal{U}(d)$.

%, and discuss their B-AC and AB-C PPT entanglement. %in view of the theory of {\it topological quantum group}. 

\subsection{Tripartite Werner states}\label{sec-Tri-Werner}

Let $\pi_A$, $\pi_{BC}$ be unitary representations of the unitary group $\mathcal{U}_d$ given by $\pi_A(U)=\overline{U}$ and $\pi_{BC}(U)=U\otimes U$. Then the elements in $\text{InvQS}(\overline{\pi_A}\otimes \pi_{BC})$ are called {\it tripartite Werner states}. Let us write $\text{Inv}(U^{\otimes 3})=\text{Inv}(\overline{\pi_A}\otimes \pi_{BC})$ and $\text{Cov}(\overline{U},UU)=\text{Cov}(\pi_A,\pi_{BC})$ for simplicity. The application of Schur-Weyl duality \cite{EW01} or von Neumann's bicommutant theorem \cite[Theorem 7.15]{wat2018} implies that the space $\text{Inv}(U^{\otimes 3})$ is spanned by six unitary operators $\left \{V_\sigma : \sigma\in \mathcal{S}_3\right\}$. Here, $V_{\sigma}:(\Comp^d)^{\otimes 3}\rightarrow (\Comp^d)^{\otimes 3}$ is determined by $V_{\sigma}(\xi_1\otimes \xi_2\otimes \xi_3)=\xi_{\sigma^{-1}(1)}\otimes \xi_{\sigma^{-1}(2)}\otimes \xi_{\sigma^{-1}(3)}$ for any $\xi_1,\xi_2,\xi_3\in \Comp^d$ and $\sigma\in \mathcal{S}_3$, or equivalently,
\begin{equation} \label{eq-UUUbasis}
    V_\sigma=\sum_{j_1,j_2,j_3=1}^d |j_1j_2j_3\ra\la j_{\sigma(1)}j_{\sigma(2)}j_{\sigma(3)}|.
\end{equation}

Recall that A-BC PPT property and separability of $\rho\in \text{InvQS}(U^{\otimes 3})$ were already characterized  %by Theorem 12 and Theorem 7 
in \cite{EW01}, and it was shown that PPT=SEP if and only if $d=2$. Therefore, a direct application of Corollary \ref{cor-CovPPTES} gives us the following result.

\begin{theorem}
All positive $(UU,\overline{U})$-covariant maps are decomposable if and only if $d= 2$. By taking the adjoint operation $\Le\mapsto \Le^*$, the same conclusion holds for positive $(\overline{U},UU)$-covariant maps.
\end{theorem}

%{\color{red} Investigate and add the use of $(U,UU)$-covariant maps, regarding the cloning process (REF: NPR 22).}

In the remaining of this section, we will assume $d\geq 3$ and exhibit positive non-decomposable $(\overline{U},UU)$-covariant maps.
%Moreover, we show that
%\begin{equation}
%    \rho=\frac{1}{d^3+5d^2+8}\left ( \frac{d+4}{d}V_e+V_{(13)}+\frac{4}{d}V_{(123)}+\frac{4}{d}V_{(132)} \right )
%\end{equation}
%is an A-BC PPT entangled quantum state in $\text{InvQS}(U^{\otimes 3})$.

[{\bf Step 1}] First of all, let us characterize all elements in $\text{CovPos}(\overline{U},UU)$. Note that Corollary \ref{cor-twirling} (3) implies that the space $\text{Cov}(\overline{U},UU)$ is spanned by the following six linear maps $\mathcal{L}_\sigma$ whose \textit{unnormalized} Choi matrices are the operators $V_\sigma\in \text{Inv}(U^{\otimes 3})$ in \eqref{eq-UUUbasis}:
\begin{equation} \label{eq-UUUCovbasis}
    \left \{\begin{array}{lllllll}
    \mathcal{L}_e(X)=(\text{Tr\,}X)\cdot  \text{Id}_d\otimes \text{Id}_d,\\
    \mathcal{L}_{(12)}(X)=X^T \otimes \text{Id}_d,\\
    \mathcal{L}_{(13)}(X)=\text{Id}_d\otimes X^T,\\
    \mathcal{L}_{(23)}(X)=(\text{Tr}\,X)\cdot \sum_{j_2,j_3=1}^d |j_3j_2\ra\la j_2j_3|,\\
    \mathcal{L}_{(123)}(X)=\sum_{j_1,j_2,j_3=1}^d X_{j_1j_2}|j_2j_3\ra\la j_3j_1|,\\
    \mathcal{L}_{(132)}(X)=\sum_{j_1,j_2,j_3=1}^d X_{j_1j_3}|j_2j_3\ra\la j_1j_2|. \end{array} \right.
\end{equation}

\begin{lemma} \label{lem-UUUpos}
Let $\Le=\sum_{\sigma\in \mathcal{S}_3} a_{\sigma}\mathcal{L}_{\sigma}\in \text{Cov}(\overline{U},UU)$. Then $\Le$ is positive if and only if
\begin{equation}\label{eq-UUUpos1}
\left \{\begin{array}{ll}
(1) &a_{e}, a_{(12)}, a_{(13)}, a_{(23)}\in\mathbb{R}\; \text{ and }\; a_{(132)}=\overline{a_{(123)}},\\
(2) &a_{e}\ge \max \left\{-a_{(12)}, -a_{(13)}, |a_{(23)}|\right\},\\
(3) &a_{e}+a_{(12)}+a_{(13)}+a_{(23)}+a_{(123)}+a_{(132)}\ge 0,\\
(4) &\left(a_{e}+a_{(12)}\right)\left(a_{e}+a_{(13)}\right)\geq \left|a_{(23)}+a_{(123)}\right|^2.
\end{array} \right.
\end{equation}
\end{lemma}
\begin{proof}
Since every unit vector $\xi\in \Comp^d$ can be written as $|\xi\ra=\overline{U}|1\ra$ for some $U\in \mathcal{U}_d$, the $(\overline{U},UU)$-covariance property implies that $\Le$ is positive if and only if $\Le(e_{11})\geq 0$. Moreover, $\Le(e_{11})$ has a matrix decomposition \small
\begin{align}
    \Le(e_{11})&\cong (a_e+a_{(12)}+a_{(13)}+a_{(23)}+a_{(123)}+a_{(132)})1\oplus (a_e+a_{(23)}) \text{Id}_{d-1} \nonumber \\
    &\oplus \left(\bigoplus_{j=2}^d \begin{bmatrix}a_e+a_{(12)} & a_{(23)}+a_{(123)} \\ a_{(23)}+a_{(132)} & a_e+a_{(13)}\end{bmatrix}\right)\oplus \left(\bigoplus_{2\leq i<j\leq d}\begin{bmatrix}a_e & a_{(23)} \\ a_{(23)} & a_e\end{bmatrix}\right)
\end{align}
\normalsize with respect to the bases $\set{|11\ra}$, $\set{|22\ra,|33\ra, \ldots, |dd\ra}$, $\set{|1j\ra,|j1\ra}$ for $j=2,\ldots, d$, and $\set{|ij\ra,|ji\ra}$ for $2\leq i<j\leq d$, respectively. Therefore, $\Le(e_{11})\geq 0$ if only if \eqref{eq-UUUpos1} holds.
\end{proof}

The next step is to classify CP and CCP conditions in $\text{Cov}(\overline{U},UU)$ to find all PPT elements in $\text{InvQS}(U^{\otimes 3})$.

\begin{lemma} \label{lem-UUUPPT}
Let $\Le=\sum_{\sigma} a_{\sigma}\mathcal{L}_{\sigma}$ and let $X=\sum_{\sigma}a_{\sigma}V_{\sigma}$. Then
\begin{enumerate}
\item $\Le$ is CP if and only if $X\geq 0$ if and only if
\begin{equation}\label{eq-UUUCP}
\left \{\begin{array}{lllll}
a_e,a_{(12)},a_{(13)},a_{(23)}\in \Real\, \text{ and }\, a_{(123)}=\overline{a_{(132)}},\\
a_e+a_{(123)}+a_{(132)}\geq |a_{(12)}+a_{(13)}+a_{(23)}|,\\
2a_{e}-a_{(123)}-a_{(132)}\geq 0,\\
(a_{e}+\omega a_{(123)}+\overline{\omega}a_{(132)})(a_e+\overline{\om}a_{(123)}+\om a_{(132)})\\
\quad\quad \ge\left\vert \omega a_{(12)}+ \overline{\om}a_{(13)}+a_{(23)}\right\vert^2.
\end{array} \right .
\end{equation}
\item $\Le$ is CCP if and only if $X^{T_A}\geq 0$ if and only if
\begin{equation}\label{eq-UUUcoCP}
\left \{\begin{array}{lllll}
a_e,a_{(12)},a_{(13)},a_{(23)}\in \Real\, \text{ and }\, a_{(123)}=\overline{a_{(132)}},\\
a_{e}\geq |a_{(23)}|,\\
2a_{e}+a_{(123)}+a_{(132)}+d(a_{(12)}+a_{(13)})\geq 0,\\
\left(a_{e}+ a_{(23)}+\frac{d+ 1}{2}(a_{(12)}+a_{(13)}+ a_{(123)}+ a_{(132)})\right)\\
\quad\quad\times \left(a_{e}- a_{(23)}+\frac{d- 1}{2}(a_{(12)}+a_{(13)}- a_{(123)}- a_{(132)})\right)\\
\quad\quad \geq \frac{d^2-1}{4}(|a_{(12)}-a_{(13)}|^2+|a_{(123)}-a_{(132)}|^2).
\end{array} \right .
\end{equation}
\end{enumerate}
\end{lemma}

These characterizations were already known from \cite[Lemma 2 and Lemma 8]{EW01}, but with a different parametrization. An elaboration on Lemma \ref{lem-UUUPPT} is attached in Appendix C using the following identifications
\begin{align}
\label{eq50}&\text{span}\left \{V_\sigma: \sigma\in \mathcal{S}_3\right\}\cong \Comp\oplus \Comp\oplus M_2(\Comp),\\
\label{eq51}&\text{span}\left \{V_\sigma^{T_A}: \sigma\in \mathcal{S}_3\right\}\cong \Comp\oplus \Comp\oplus M_2(\Comp)
\end{align}
as $*$-algebras.
%The above decompositions are thanks to the standard representation theory of the unitary group $\mathcal{U}_d$ and \eqref{eq-fusion3}.

%For general $d\in \n$,
%\begin{align}
%R_+&=\frac{1}{6}\left (V_{e} + V_{(12)} + V_{(13)} + V_{(23)} + V_{(123)} + V_{(132)} \right ),\\
%R_-&=\frac{1}{6}\left (V_{e} -  V_{(12)} - V_{(13)} - V_{(23)} + V_{(123)} + V_{(132)} \right ),\\
%R_0&=\frac{1}{3}\left (2\cdot V_e - V_{(123)} - V_{(132)} \right )
%\end{align}
%are orthogonal projections satisfying $R_++R_-+R_0=\text{id}_{d^3}$ on $M_d(\Comp)^{\otimes 3}$.

{[\bf{Step 2}]} All extremal elements in $\text{CovPosTP}(\overline{U},UU)$ are completely characterized in the following lemma. Our proof is straightforward but rather cumbersome, so we attach the proof in Appendix \ref{sec-extpos}.

\begin{lemma} \label{lem-UUUext}
Let $\Le=\sum_{\sigma}a_{\sigma}\Le_{\sigma}\in \text{Cov}(\overline{U},UU)$. Then the following are equivalent.
\begin{enumerate}
\item $\Le\in \text{Ext}(\text{CovPosTP}(\overline{U},UU))$
\item $a_e,a_{(12)},a_{(13)},a_{(23)}\in \mathbb{R}$, $a_{(123)}=\overline{a_{(132)}}$, and the associated 6-tuple 
\begin{equation}
(a_e,a_{(12)},a_{(13)},a_{(23)},\text{Re}(a_{(123)}),\text{Im}(a_{(123)}))\in \mathbb{R}^6
\end{equation}
is one of the following three types:
\[\begin{array}{ll}
\text{Type \MakeUppercase{\romannumeral 1}}&c_1(1,-1,-1,-1,1,0)\\
\text{Type \MakeUppercase{\romannumeral 2}}&c_2(0,A,B,0,C,\pm\sqrt{AB-C^2})\\
\text{Type \MakeUppercase{\romannumeral 3}} &c_3(\frac{A+B+2C}{2}, \frac{A-B-2C}{2}, \frac{-A+B-2C}{2}, \frac{A+B+2C}{2}, -\frac{A+B}{2},\pm\sqrt{AB-C^2})
\end{array}\]
where $A,B\geq 0$, $C\in \Real$, $AB\geq C^2$, and the normalizing constants $c_1,c_2$, and $c_3$ are chosen to satisfy the TP condition
\begin{equation}
    d^2a_e+d(a_{(12)}+a_{(13)}+a_{(23)})+(a_{(123)}+a_{(132)})=1.
\end{equation}
\end{enumerate}
\end{lemma}

Then, combining Lemma \ref{lem-UUUPPT} and Lemma \ref{lem-UUUext}, we can check that
\begin{itemize}
    \item Every $\Le\in \text{Ext}(\text{CovPosTP}(\overline{U},UU))$ of Type \MakeUppercase{\romannumeral 1} is CP,
    
    \item Every $\Le\in \text{Ext}(\text{CovPosTP}(\overline{U},UU))$ of Type \MakeUppercase{\romannumeral 2} is CCP,
    
    \item Let $\Le\in \text{Ext}(\text{CovPosTP}(\overline{U},UU))$ of Type \MakeUppercase{\romannumeral 3}. Then 
    \begin{itemize}
    \item $\Le$ is CP if and only if $A=B=C$, 
    \item $\Le$ is CCP if and only if $A=B=-C$.
    \end{itemize}
\end{itemize}
Thus, Type \MakeUppercase{\romannumeral 3} (with neither $A=B=C$ nor $A=B=-C$) provides explicit positive non-decomposable maps in $\text{CovPos}(\overline{U},UU)$ by Theorem \ref{thm-extremal-decomposable}. For example, we can choose $A=1,B=0$, and $C=0$ to obtain a specific extremal element
\begin{equation}
    \Le_0 = \Le_e+\Le_{(12)}-\Le_{(13)}+\Le_{(23)}-\Le_{(123)}-\Le_{(132)}\in \text{Ext}(\text{CovPosTP}(\overline{U},UU))
\end{equation}
up to a normalizing constant.

[{\bf Step 3}] On the dual side, the chosen positive non-decomposable map $\mathcal{L}_0^*\in \text{CovPos}(UU,\overline{U})$ should play a role as a PPT entanglement detector. Indeed, if we take
\begin{equation}
    \rho_t=\frac{1}{d^3+(t+1)d^2+2t}\left (\frac{d+t}{d}V_e+V_{(13)}+\frac{t}{d}V_{(123)}+\frac{t}{d}V_{(132)}\right )
\end{equation} 
with $0< t\leq 3.89$ and $d\geq 3$, then $\rho_t\in \text{InvQS}(U^{\otimes 3})$ is A-BC PPT by Lemma \ref{lem-UUUPPT}. Moreover,
%for the given $\Le_0\in \text{CovPos}_1(\overline{U},UU)$ and $\rho_t \in \text{InvQS}(U^{\otimes 3})$,
it is straightforward to see that \small
\begin{align}
&(d^3+(t+1)d^2+2t)\cdot (\text{id}\otimes \Le_0^*)(\rho_t) \nonumber \\
&=\left(d^{2}+(t+2)d+3t-\frac{2t}{d}\right)\text{Id}_{d}\otimes\text{Id}_{d}-\left(d^{2}+(2t+2)d\right)|\Omega_{d}\ra\la\Omega_{d}|
\end{align}
\normalsize has a negative eigenvalue $\displaystyle -t\left(d+\frac{2}{d}-3\right)<0$. Consequently, the quantum state $\rho_t$ is A-BC PPT entangled by Theorem \ref{thm-main} or by Horodecki's criterion. 

\begin{remark}
Note that $\rho_t$ is also C-AB PPT entangled since $V_{(13)}\rho_t V_{(13)}=\rho_t$. On the other hand, $\rho_t$ is not B-AC PPT (and hence entangled). Indeed, we can observe that
\begin{equation}
    \rho_t^{T_B}=V_{(12)}(V_{(12)}\rho_t V_{(12)})^{T_A}V_{(12)},
\end{equation}
but $V_{(12)}\rho_t V_{(12)}$ is not A-BC PPT since \small
\begin{equation} \label{eq-B-ACnonPPT}
V_{(12)}\rho_t V_{(12)}=\frac{1}{d^3+(t+1)d^2+2t}\left (\frac{d+t}{d}V_e+V_{(23)}+\frac{t}{d}V_{(123)}+\frac{t}{d}V_{(132)}\right)
\end{equation}
\normalsize does not satisfy the CCP condition \eqref{eq-UUUcoCP}. 

It might be interesting if we can find a tripartite PPT-entangled Werner state with respect to all the three partitions A-BC, B-AC, and C-AB. However, Lemma 7 of \cite{EW01} implies that there is no such an example $\rho=\sum_{\sigma}a_{\sigma}V_{\sigma}$ if one of the following conditions is satisfied:
\begin{itemize}
    \item $a_{(12)}=a_{(13)}$ and $a_{(123)}=a_{(132)}$,
    
    \item$a_{(13)}=a_{(23)}$ and $a_{(123)}=a_{(132)}$,
    
    \item $a_{(23)}=a_{(12)}$ and $a_{(123)}=a_{(132)}$,
\end{itemize}
%which means that $\rho$ has ``little'' symmetry among $\text{InvQS}(U^{\otimes 3})$. At this point, 
We leave the general situation as an open question. %{@@@@\color{red} (ADD reference that such examples exist in general class of tripartite matrices)}

\end{remark}

\subsection{Tripartite quantum orthogonally invariant quantum states}\label{sec:Q.orthogonal}

Within the framework of compact quantum groups, it is well-known that the orthogonal group $\mathcal{O}_d$ allows a universal object, namely the {\it free orthogonal quantum group} $\mathcal{O}_d^+$ \cite{Wa95, tim08}. In other words, the invariance property with respect to $\mathcal{O}_d^+$ is a stronger notion than the (classical) orthogonal group invariance. See \cite{LeYo22} for a general discussion on invariant quantum states and covariant quantum channels with quantum group symmetries.

In this section, we focus on the space $\text{Inv}(O_+^{\otimes 3})%=\text{span}\left\{T_e,T_{(12)},T_{(23)},T_{(123)},T_{(132)}\right\}
$ of the tripartite {\it quantum orthogonally invariant operators} spanned by five tripartite operators
\begin{equation} \label{eq-OOOinv}
    T_{\sigma}=V_{\sigma}^{T_B}=\sum_{j_1,j_2,j_3=1}^d|j_1j_{\sigma(2)}j_3\ra\la j_{\sigma(1)}j_2j_{\sigma(3)}|
\end{equation}
for $\sigma\in \mathcal{S}_3\setminus \set{(13)}=\set{e,(12),(23),(123),(132)}$. See Appendix \ref{sec:FOQG} for more discussions on \eqref{eq-OOOinv} and $\mathcal{O}_d^+$. Although Theorem \ref{thm-main} does not cover quantum group symmetries, any $X\in \text{Inv}(O_+^{\otimes 3})$ satisfies the following group invariance property
\begin{equation}
    (U\otimes \overline{U}\otimes U)X(U\otimes \overline{U}\otimes U)^*=X
\end{equation}
for all $U\in \mathcal{U}_d$ thanks to Corollary \ref{cor-twirling} (1). This transfers our problem to the realm of classical group symmetries. More precisely, we have
\begin{equation}
    \text{Inv}(O_+^{\otimes 3})\subseteq \text{Inv}(U\otimes \overline{U}\otimes U)=\text{Inv}(\overline{\pi_A}\otimes \pi_{BC})
\end{equation}
where $\pi_A(U)=\overline{U}$ and $\pi_{BC}(U)=\overline{U}\otimes U$. The main theorem of this section is the following.

\begin{theorem} \label{thm-qOOO}
Let $\rho\in \text{InvQS}(U\otimes \overline{U}\otimes U)$. Then $\rho$ is A-BC separable if and only if $\rho$ is A-BC PPT. In particular, A-BC PPT= A-BC SEP holds in $\text{InvQS}(O^{\otimes 3}_+)$.
\end{theorem}

{
\begin{remark} \label{rmk-COS18}
Theorem \ref{thm-qOOO} implies POS=DEC in $\text{Cov}(U,U\overline{U})$ by Corollary \ref{cor-CovPPTES}, and equivalently, POS=DEC in $\text{Cov}(U,\overline{U}U)$. We remark that the latter class was analyzed recently in \cite{COS18,BCS20}. In particular, $k$-positivity and decomposability were discussed for a special subclass of $\text{Cov}(U,\overline{U}U)$ for $d=3$ in \cite{COS18}, and it was questioned whether certain $2$-positive non-decomposable map exists in $\text{Cov}(U,\overline{U}U)$. However, our Theorem \ref{thm-qOOO} gives the complete answer POS=DEC in the whole class $\text{Cov}(U,\overline{U}U)$ regardless of the dimension $d$, and this means that ($k$-)positive non-decomposable maps cannot exist in $\text{Cov}(U,\overline{U}U)$.
%Theorem \ref{thm-qOOO} is equivalent to POS=DEC in $\text{Cov}(\overline{U},\overline{U}U)=\text{Cov}(U,U\overline{U})$ by Corollay \ref{cor-CovPPTES} and also equivalent to POS=DEC in $\text{Cov}(U,\overline{U}U)$. The latter class has been recently analyzed \cite{COS18,BCS20} under the term \textit{$(1,1)$-unitarily equivariant maps}. In particular, $k$-positivity and decomposability of a subclass of $\text{Cov}(U,\overline{U}U)$ are considered in \cite{COS18}, and it is conjectured that there exists a $2$-positive nondecomposable map in this subclass when $d=3$. Therefore, our result gives much generalized (negative) answer in the sense that it is impossible to find any ($k$-)positive nondecomposable maps in the whole class of $(1,1)$-unitarily equivariant maps, regardless of the dimension $d$.
\end{remark}
}

\begin{remark}
Note that, for any unitary representation $\pi$ of a compact group, the following three problems
\begin{itemize}
\item A-BC PPT = A-BC SEP in $\text{Inv}(\pi\otimes \pi\otimes \pi)$
\item B-AC PPT = B-AC SEP in $\text{Inv}(\pi\otimes \pi\otimes \pi)$
\item C-AB PPT = C-AB SEP in $\text{Inv}(\pi\otimes \pi\otimes \pi)$
\end{itemize}
are equivalent. However, Theorem \ref{thm-qOOO} implies that this equivalence is no longer true when $\pi$ is replaced by a unitary representation of a compact quantum group. Indeed,  a B-AC PPT entangled state $V_{(12)}\rho_t V_{(12)}\in \text{InvQS}(U^{\otimes 3})$ from \eqref{eq-B-ACnonPPT} is transferred to the following B-AC PPT entangled state in $\text{InvQS}(O_+^{\otimes 3})$:
\begin{align}
    &(V_{(12)}\rho_t V_{(12)})^{T_B} \nonumber \\
    &=\frac{1}{d^3+(t+1)d^2+2t} \left( \frac{d+t}{d} T_{e}+T_{(23)}+\frac{t}{d}T_{(123)}+\frac{t}{d}T_{(132)}\right ).
\end{align}
In other words, the problem of A-BC PPT=SEP is not equivalent to the problem of B-AC PPT=SEP in $\text{InvQS}(O_+^{\otimes 3})$. A reason for this genuine quantum phenomenon is that the associated $C^*$-algebra of $\mathcal{O}_d^+$ is non-commutative.
\end{remark}

{[\bf{Step 1}]} Let us apply Corollary \ref{cor-equiv} to prove that POS=DEC holds in $\text{CovPos}_1(\overline{U}U,\overline{U})=\text{CovPos}_1(\pi_{BC},\pi_A)$, or equivalently, POS=DEC holds in $\text{CovPosTP}(\overline{U}, \overline{U}U)$. Recall that the space $\text{Cov}(\overline{U}, \overline{U}U)$ is spanned by six linear maps
\begin{equation} \label{eq-OOOcov}
    \mathcal{M}_{\sigma}=(T_d\otimes \text{id}_d)\circ \Le_{\sigma}
\end{equation}
for $\sigma\in \mathcal{S}_3$, where $\Le_{\sigma}$ is given by \eqref{eq-UUUCovbasis}. Then the unnormalized Choi matrix of $\mathcal{M}_{\sigma}$ is given by $T_{\sigma}$. 

\begin{lemma} \label{lem-UUbarUpos}
Let $d\geq 3$ and $\mathcal{M}=\sum_{\sigma}a_{\sigma}\mathcal{M}_{\sigma}\in \text{Cov}(\overline{U},\overline{U}U)$. Then $\mathcal{M}$ is positive if and only if
\begin{equation}\label{eq-UUbarUpos}
\left\{\begin{array}{ll}
(1)&a_{e}, a_{(12)}, a_{(13)}, a_{(23)}\in\mathbb{R}\;\text{ and }\; a_{(132)}=\overline{a_{(123)}},\\
(2)& a_e\geq \max\set{0,-a_{(12)},-a_{(13)}}, \\
(3)& a_e+a_{(12)}+a_{(13)}+a_{(23)}+a_{(123)}+a_{(132)}\geq 0,\\
(4)& a_{e}+(d-1)a_{(23)}\geq 0,\\
(5)& (a_e+a_{(12)}+a_{(13)}+a_{(23)}+a_{(123)}+a_{(132)})(a_e+(d-1)a_{(23)})\\
&\quad \geq (d-1)|a_{(23)}+a_{(123)}|^2.
\end{array}\right.
\end{equation}
\end{lemma}

\begin{proof}
As in the proof of Lemma \ref{lem-UUUpos}, the positivity of $\mathcal{M}$ is equivalent to
$\mathcal{M}(e_{11})\geq 0$. Moreover, $\mathcal{M}(e_{11})\in M_d(\Comp)\otimes M_d(\Comp)$ has a matrix decomposition
\begin{equation}\label{eq-decomp}
    %\mathcal{M}(e_{11})\cong 
    M \oplus (a_e+a_{(12)})\,\text{Id}_{d-1} \oplus (a_e+a_{(13)})\,\text{Id}_{d-1} \oplus a_e\,\text{Id}_{(d-1)(d-2)},
\end{equation}
where
\begin{equation} \label{eq-UUbarUmatrix}
    M= \begin{bmatrix} c & (a_{(23)}+a_{(132)}) \la v| \\ (a_{(23)}+a_{(123)}) |v\ra & a_{(23)} |v\ra\la v| \end{bmatrix}+a_e\, \text{Id}_d
\end{equation}
with $c=a_{(12)}+a_{(13)}+a_{(23)}+a_{(123)}+a_{(132)}$ and  $v=(1,1,\ldots, 1)^T\in \Comp^{d-1}$. The four matrices in the matrix decomposition \eqref{eq-decomp} are with respect to the bases $\set{|11\ra, |22\ra, \ldots, |dd\ra}$, $\set{|12\ra, |13\ra, \ldots, |1d\ra}$, $\set{|21\ra, |31\ra, \ldots, |d1\ra}$, and $\set{|ij\ra: i,j\neq 1 \text{ and } i\neq j}$, respectively. Thus, $\mathcal{M}(e_{11})\geq 0$ if and only if the conditions (1) and (2) in \eqref{eq-UUbarUpos} hold and $M\geq 0$.  Moreover, we can rewrite \eqref{eq-UUbarUmatrix} as
\begin{equation}
    M-a_e\,\text{Id}_d=V\begin{bmatrix} c & \sqrt{d-1}\overline{\alpha} \\ \sqrt{d-1}\alpha & (d-1)a_{(23)} \end{bmatrix}V^*,
\end{equation}
where $\alpha=a_{(23)}+a_{(123)}$ and $V=\begin{bmatrix} 1 & 0 \\ 0 & \frac{1}{\sqrt{d-1}}|v\ra \end{bmatrix}\in M_{d,2}(\Comp)$ is an isometry. Thus, the nonzero eigenvalues of $M-a_e\text{Id}_d$ are the same with those of $\begin{bmatrix} c & \sqrt{d-1}\overline{\alpha} \\ \sqrt{d-1}\alpha & (d-1)a_{(23)} \end{bmatrix}$. 
Consequently, the condition $M\geq 0$ is equivalent to the conditions (3), (4), and (5) of \eqref{eq-UUbarUpos}.
\end{proof}

Thanks to Lemma \ref{lem-UUUPPT}, it is easy to derive CP and CCP conditions in $\text{Cov}(\overline{U},\overline{U}U)$ or A-BC PPT condition in $\text{InvQS}(U\otimes \overline{U}\otimes U)$.

\begin{lemma}\label{lem-UUbarUPPT}
Let $d\geq 3$ and $\mathcal{M}=\sum_{\sigma\in \mathcal{S}_3}a_{\sigma}\mathcal{M}_{\sigma}\in \text{CovPos}(\overline{U},\overline{U}U)$. Then
\begin{enumerate}
    \item $\mathcal{M}$ is CP if and only if the operator 
    \begin{equation}
        V_{(12)}\left(\sum_{\sigma\in \mathcal{S}_3} a_{\sigma}V_{\sigma}\right)V_{(12)}=\sum_{\sigma\in \mathcal{S}_3} a_{(12) \sigma  (12)}V_{\sigma}
    \end{equation}
    satisfies the condition \eqref{eq-UUUcoCP}.
    
    \item $\mathcal{M}$ is CCP if and only if the operator
    \begin{equation}
        V_{(13)}\left(\sum_{\sigma\in \mathcal{S}_3} a_{\sigma}V_{\sigma}\right)V_{(13)}=\sum_{\sigma\in \mathcal{S}_3} a_{(13)\sigma(13)}V_{\sigma}
    \end{equation}
    satisfies the condition \eqref{eq-UUUcoCP}.
\end{enumerate}
\end{lemma}

\begin{proof}
Let $X=\sum_{\sigma\in \mathcal{S}_3}a_{\sigma}V_{\sigma}\in \text{Inv}(U^{\otimes 3})$ and $X'=V_{(12)}X V_{(12)}$. Then $\mathcal{M}$ is CP if and only if 
\begin{equation}
    \sum_{\sigma\in \mathcal{S}_3} a_{\sigma}T_{\sigma}=X^{T_B}=V_{(12)}(X')^{T_A}V_{(12)} \geq 0,
\end{equation}
which is equivalent to $(X')^{T_A}\geq 0$. On the other hand, let $X''=V_{(13)}X V_{(13)}$. Then $\mathcal{M}$ is CCP if and only if 
\begin{equation}
    \left(\sum_{\sigma\in \mathcal{S}_3} a_{\sigma}T_{\sigma}\right)^{T_A}=\left(X^{T_C}\right)^T=\left(V_{(13)}(X'')^{T_A}V_{(13)}\right)^T\geq 0,
\end{equation}
and this is equivalent to $(X'')^{T_A}\geq 0$.
\end{proof}

{[\bf{Step 2}]} We refer the reader to Appendix \ref{sec-extpos} for a proof of the following Lemma \ref{lem-UUbarUext} classifying all extremal elements in $\text{CovPosTP}(\overline{U},\overline{U}U)$. 
\begin{lemma} \label{lem-UUbarUext}
Let $d\geq 3$ and $\mathcal{M}=\sum_{\sigma\in \mathcal{S}_3}a_{\sigma}\mathcal{M}_{\sigma}\in \text{CovPos}(\overline{U},\overline{U}U)$. Then the following are equivalent.
\begin{enumerate}
\item $\M\in \text{Ext}(\text{CovPosTP}(\overline{U},\overline{U}U))$,
\item $a_e,a_{(12)},a_{(13)},a_{(23)}\in \mathbb{R}$, $a_{(123)}=\overline{a_{(132)}}$, and the associated 6-tuple 
\begin{equation}
(a_e,a_{(12)},a_{(13)},a_{(23)},\text{Re}(a_{(123)}),\text{Im}(a_{(123)}))\in \mathbb{R}^6
\end{equation}
is one of the following four types:

\[\begin{array}{ll}
\text{Type \MakeUppercase{\romannumeral 1}}&c_1(d-1,-1,1-d,-1,1,0),\\
\text{Type \MakeUppercase{\romannumeral 2}}&c_2(d-1,1-d,-1,-1,1,0),\\
\text{Type \MakeUppercase{\romannumeral 3}}&c_3(0,A+B-2C,0,B,C-B,\pm \sqrt{AB-C^2}),\\
\text{Type \MakeUppercase{\romannumeral 4}}&c_4(0,0,A+B-2C,B,C-B,\pm \sqrt{AB-C^2}),
\end{array}\]

where $A,B\geq 0$, $C\in \Real$, $AB\geq C^2$, and $c_i$ ($i=1,2,3,4$) are normalizing constants chosen to satisfy the TP condition
\begin{equation}
    d^2a_e+d(a_{(12)}+a_{(13)}+a_{(23)})+(a_{(123)}+a_{(132)})=1.
\end{equation}
\end{enumerate}
\end{lemma}

\begin{proof} [\textbf{Proof of Theorem \ref{thm-qOOO}}]
Let us assume $d\geq 3$. According to Theorem \ref{thm-extremal-decomposable} and Lemma \ref{lem-UUbarUext}, it suffices to show that $\mathcal{M}=\sum_{\sigma\in \mathcal{S}_3}a_{\sigma}\mathcal{M}_{\sigma}$ is CP or CCP whenever $(a_{\sigma})_{\sigma\in \mathcal{S}_3}$ is one of the four Types \MakeUppercase{\romannumeral 1} - \MakeUppercase{\romannumeral 4}. Now, by applying Lemma \ref{lem-UUbarUPPT}, we can check that

\begin{itemize}
\item Every $\Le\in \text{Ext}(\text{CovPosTP}(\overline{U},\overline{U}U))$ of Type \MakeUppercase{\romannumeral 1} or Type \MakeUppercase{\romannumeral 3} is CP, 
\item Every $\Le\in \text{Ext}(\text{CovPosTP}(\overline{U},\overline{U}U))$ of Type \MakeUppercase{\romannumeral 2} or Type \MakeUppercase{\romannumeral 4} CCP,
\end{itemize}
thanks to the conditions $A,B\geq 0$ and $AB\geq C^2$. When $d=2$, we refer to Appendix \ref{sec-d=2} for the complete proof.
\end{proof}

%\begin{remark}
%{\color{red} How to modify??}
%One might wonder if Theorem \ref{thm-main} extends to more general contexts of compact quantum groups. The above discussion can be considered an affirmative case study particularly for the standard quantum orthogonal symmetries since the entanglement detector $\mathcal{M}$ in \eqref{eq51} is chosen to be in $\text{CovPos}(O_+,O_+^{\otimes 2})$.
%\end{remark}

\emph{Acknowledgements}: The authors thank Professor Hun Hee Lee for the helpful discussions and comments. S-J.Park, Y-G.Jung and S-G.Youn were supported by the National Research Foundation of Korea (NRF) grant funded by the Ministry of Science and ICT (MSIT) (No.2021K1A3A1A21039365). Y-G.Jung, J.Park and S-G.Youn were also supported by Samsung Science and Technology Foundation under Project Number SSTF-BA2002-01. S-J.Park and S-G.Youn were supported by the National Research Foundation of Korea (NRF) grant funded by the Ministry of Science and ICT (MSIT) (No. 2020R1C1C1A01009681).

\appendix

\section{Covariant positive maps with respect to monomial unitary groups} \label{sec-CovMU}

In Section 5 and Section 6 of \cite{KMS20}, the authors analyzed the general structure of irreducibly covariant linear maps under some natural symmetries of the symmetric group $\mathcal{S}_4$ and the monomial unitary group $\mathcal{MU}(d,n)$, and presented new examples of positive irreducibly covariant maps. In this section, we elaborate on how our Theorem \ref{prop-HHpos} strengthens their results and resolves several open questions raised in \cite{KMS20}.

On one side, they considered irreducibly $(\tau_3,\tau_3)$-covariant maps, where $\tau_3:\mathcal{S}_4\to \mathcal{U}_3$ is a $3$-dimensional irreducible component of %the standard representation of the symmetric group $\mathcal{S}_4$ which corresponds to the partition $(3,1)$ (REF: Sec 5.1). More precisely, 
the canonical representation $\sigma\in \mathcal{S}_4\mapsto \sum_{k=1}^4 |\sigma(k)\ra\la k|\in \mathcal{U}(4)$. More precisely, this canonical representation is not irreducible; it allows an invariant $1$-dimensional subspace $\Comp\cdot |v\ra$ with $|v\ra=\sum_{k=1}^d|k\ra$ and the other $3$-dimensional invariant subspace $V=(\Comp\cdot |v\ra)^{\perp}$. Then the fundamental representation $\tau_3:\mathcal{S}_4\rightarrow \mathcal{U}(3)$ is defined by $\tau_3(\sigma)=\Pi_V \sigma \big|_{V}\in \mathcal{U}(3)$ for all $x\in \mathcal{S}_4$, where $\Pi_V$ is the orthogonal projection from $\Comp^4$ onto $V$.

The authors characterized all $(\tau_3,\tau_3)$-covariant maps and suggested a sufficient condition for positivity using the  the so-called {\it inverse reduction map criterion} \cite{MRS15}. On the other hand, it was shown in \cite[Section 6.1.1]{LeYo22} that, up to a change of basis, the $(\tau_3,\tau_3)$-covariant maps are precisely the linear combinations of $\Psi_1,\Psi_2,\Psi_3,\Psi_4:M_3(\Comp)\to M_3(\Comp)$ from \eqref{eq-HHCPext} with $d=3$. In other words, we have $\text{Cov}(\tau_3,\tau_3)=\text{Cov}(H,H)$ up to a unitary equivalence. Thus, Theorem \ref{prop-HHpos} gives the complete solution to the open question of the characterization of all positive $(\tau_3,\tau_3)$-covariant maps raised in \cite{KMS20}.

On the other side, %they analyzed covariance maps with respect to subgroups of the monomial unitary group. 
recall that the \textit{monomial unitary group} $\mathcal{MU}(d)$ is a subgroup of $\mathcal{U}_d$ generated by all permutation matrices and all diagonal unitary matrices. Moreover, the subgroup $\mathcal{MU}(d,n)$ of $\mathcal{MU}(d)$ is generated by all permutation matrices and all diagonal matrices of the form $\sum_{i=1}^d \om_i |i\ra\la i|$ where $\om_i\in \set{1, e^{2\pi i/n},\ldots, e^{2(n-1)\pi i/n}}$, In particular, we have $\mathcal{MU}(d,2)=\mathcal{H}_d$.

For a closed subgroup $G$ of $\mathcal{U}_d$, we denote by $\pi_G:x\in G\mapsto x\in \mathcal{U}_d$ the fundamental representation of $G$, temporarily in this section. It was shown in \cite{KMS20} that, if $n\geq 3$, then 
\begin{equation}
\text{Cov}(\pi_{\mathcal{MU}(d,n)},\pi_{\mathcal{MU}(d,n)})=\text{span}\left\{\psi_0,\psi_1,\psi_3\right\}
\end{equation}
where $\psi_0,\psi_1,\psi_3$ are from \eqref{eq-HHCovbasis}. Moreover, the authors characterized all positive maps in this class and proved that all $(\pi_{\mathcal{MU}(d,n)},\pi_{\mathcal{MU}(d,n)})$-covariant positive maps are decomposable for $n\geq 3$. However, explicit decompositions were left as an open question, and the authors conjectured that a non-decomposable positive map may arise under $(\pi_{\mathcal{MU}(d)},\pi_{\mathcal{MU}(d)})$-covariance. 

Our results in this paper resolve their open questions in the sense that $(\pi_{\mathcal{MU}(d)},\pi_{\mathcal{MU}(d)})$-covariance does not make a difference, but a weaker condition $(\pi_{\mathcal{MU}(d,2)},\pi_{\mathcal{MU}(d,2)})$-covariance does. Furthermore, their POS=DEC result in $\text{Cov}(\pi_{\mathcal{MU}(d,n)},\pi_{\mathcal{MU}(d,n)})$ with $n\geq 3$ (Section 6.c of \cite{KMS20}) extends to a more general result POS=DEC in $\text{Cov}(\pi_{\mathcal{MU}(d,2)},\pi_{\mathcal{MU}(d,2)})$ with explicit decompositions into the sum of CP and CCP maps.
\begin{enumerate}
    \item More precisely, it is clear that \small
    \begin{equation}
    \text{Cov}(\pi_{\mathcal{MU}(d)},\pi_{\mathcal{MU}(d)})\subseteq \text{Cov}(\pi_{\mathcal{MU}(d,n)},\pi_{\mathcal{MU}(d,n)})=\text{span}\set{\psi_0,\psi_1,\psi_3},
    \end{equation}
    \normalsize and all the three maps $\psi_0,\psi_1$, and $\psi_3$ are covariant with respect to general diagonal unitary matrices. Therefore, for $n\geq 3$ we have \small
    \begin{align}
        \text{Cov}(\pi_{\mathcal{MU}(d)},\pi_{\mathcal{MU}(d)}) & = \text{Cov}(\pi_{\mathcal{MU}(d,n)},\pi_{\mathcal{MU}(d,n)})=\text{span}\set{\psi_0,\psi_1,\psi_3}.
    \end{align}
    \normalsize Therefore, there is no positive non-decomposable element inside $\text{Cov}(\pi_{\mathcal{MU}(d)},\pi_{\mathcal{MU}(d)})$.
    \item On the other hand, since $\mathcal{MU}(d,2)=\mathcal{H}_d$, we have \small
    \begin{align}
    \text{Cov}(\pi_{\mathcal{MU}(d,2)},\pi_{\mathcal{MU}(d,2)})&=\text{Cov}(H,H)=\text{span}\set{\psi_0,\psi_1,\psi_2,\psi_3},
    \end{align}
    \normalsize by Proposition \ref{prop41}. Moreover, POS=DEC in $\text{Cov}(\pi_{\mathcal{MU}(d)},\pi_{\mathcal{MU}(d)})$ from Section 6.c of \cite{KMS20} is strengthened to POS=DEC in the larger space $\text{Cov}(\pi_{\mathcal{MU}(d,2)},\pi_{\mathcal{MU}(d,2)})$ with explicit decompositions by Theorem \ref{prop-HHpos}.
\end{enumerate}

\section{Characterization of $\text{InvPPTQS}(U^{\otimes 3})$}
Recall that if $d\geq 3$, there exist $*$-algebra isomorphisms 
\begin{align}
&F:\text{Inv}(U\otimes U\otimes U)\rightarrow \Comp\oplus \Comp\oplus M_2(\Comp),\\
&G: \text{Inv}(\overline{U}\otimes U\otimes U) \rightarrow \Comp\oplus \Comp\oplus M_2(\Comp)
\end{align}
by the representation theory of the unitary group $\mathcal{U}_d$. Moreover, the authors in \cite{EW01} proposed specific choice of $*$-isomorphisms $F$ and $G$ that can be used to characterize the PPT condition of $X=\sum_{\sigma\in \mathcal{S}_3}a_\sigma 
V_\sigma\in \text{Inv}(U^{\otimes 3})$. For the convenience of the reader, we again present explicit maps $F$ and $G$ in terms of the bases $\set{V_{\sigma}:\sigma\in \mathcal{S}_3}$ and $\set{V_{\sigma}^{T_A}:\sigma\in \mathcal{S}_3}$ of $\text{Inv}(U^{\otimes 3})$ and $\text{Inv}(\overline{U}\otimes U\otimes U)$, respectively, in Table 2.
\small
\begin{table}[h!] \label{tab-8}
  \begin{center}
    \caption{The isomorphisms $F$ and $G$\; ($\om=e^{2\pi i/3 }$)}
    \label{tab:table8}
    \begin{tabular}{|c|c|c|c|c|c|c|c|} % <-- Alignments: 1st column left, 2nd middle and 3rd right, with vertical lines in between
       %$\backslash$ 
\hline
 $X$ &$F(X)$&$X^{T_A}$&$G(X^{T_A})$\\
  \hline
  $V_e$ &  $(1,1,\left [ \begin{array}{cc} 1&0 \\ 0&1 \end{array} \right ])$ &$V_e^{T_A}$& 
  $(1,1, \left [ \begin{array}{cc} 1&0 \\ 0&1 \end{array} \right ])$\\
\hline
$V_{(12)}$ & $(1,-1, \left [ \begin{array}{cc} 0& \overline{\omega} \\ {\omega}&0 \end{array} \right ])$  &$V_{(12)}^{T_A}$&$ (0,0,\left [ \begin{array}{cc} \frac{d+1}{2} &\frac{\sqrt{d^2-1}}{2} \\ \frac{\sqrt{d^2-1}}{2} & \frac{d-1}{2} \end{array} \right ] ) $ 
  \\
\hline
$V_{(13)}$ &  $(1,-1,\left [ \begin{array}{cc} 0&{\omega} \\ \overline{\omega} &0 \end{array} \right ])$ &$V_{(13)}^{T_A}$& 
  $(0,0,\left [ \begin{array}{cc} \frac{d+1}{2} &-\frac{\sqrt{d^2-1}}{2} \\ -\frac{\sqrt{d^2-1}}{2} & \frac{d-1}{2} \end{array} \right ] )$\\
\hline
$V_{(23)}$ & $(1,-1,\left [ \begin{array}{cc} 0&1 \\ 1&0 \end{array} \right ])$  &$V_{(23)}^{T_A}$&$(1,-1,\left [ \begin{array}{cc} 1&0 \\ 0&-1 \end{array} \right ])$
  \\
\hline
$V_{(123)}$ & $(1,1,\left [ \begin{array}{cc} \overline{\omega} & 0 \\ 0 &{\omega} \end{array} \right ])$ &$V_{(123)}^{T_A}$&
$(0,0,\left [ \begin{array}{cc} \frac{d+1}{2} &-\frac{\sqrt{d^2-1}}{2} \\ \frac{\sqrt{d^2-1}}{2} & -\frac{d-1}{2} \end{array} \right ])$
  \\
\hline
$V_{(132)}$ & $(1,1,\left [ \begin{array}{cc} {\omega} & 0 \\ 0 & \overline{\omega} \end{array} \right ])$ &$V_{(132)}^{T_A}$& $(0,0,\left [ \begin{array}{cc} \frac{d+1}{2} &\frac{\sqrt{d^2-1}}{2} \\ -\frac{\sqrt{d^2-1}}{2} &  - \frac{d-1}{2} \end{array} \right ] )$
  \\
\hline
    \end{tabular}
  \end{center}
\end{table}

\normalsize

\begin{proof}[\textbf{Proof of Lemma \ref{lem-UUUPPT}}]
Let $X=\sum_{\sigma\in \mathcal{S}_3} a_{\sigma}V_{\sigma}$. Then $X^*=X$ is equivalent to $a_e, a_{(12)},a_{(13)},a_{(23)}\in \Real$ and $a_{(123)}=\overline{a_{(132)}}$ since $\set{V_{\sigma}}_{\sigma\in \mathcal{S}_3}$ is linearly independent. Now $X\geq 0$, or equivalently, $F(X)\geq 0$ holds if and only if
\begin{align}
    (a_{e}+a_{(123)}+a_{(132)})\pm(a_{(12)}+a_{(13)}+a_{(23)})\geq 0\; \text{ and }\;\\ \begin{bmatrix} a_e+\overline{\om}a_{(123)}+a_{(123)}\om & \overline{\om}a_{(12)}+\om a_{(13)}+a_{(23)} \\ \om a_{(12)}+\overline{\om}a_{(13)}+a_{(23)} & a_{e}+\om a_{(123)}+\overline{\om}a_{(132)}  \end{bmatrix} \geq 0,
\end{align}
which is equivalent to the condition \eqref{eq-UUUCP}. Similarly, we get the equivalence between the condition $X^{T_A}\geq 0$ and \eqref{eq-UUUcoCP}.
\end{proof}

\section{Extremal positive maps in $\text{CovPosTP}(\overline{U},UU)$ and $\text{CovPosTP}(\overline{U},\overline{U}U)$} \label{sec-extpos}

This section is to give detailed proofs of Lemma \ref{lem-UUUext} and Lemma \ref{lem-UUbarUext}. For convenience, we assume $a_{e},a_{(12)},a_{(13)},a_{(23)}\in \Real$, $a_{(123)}=\overline{a_{(132)}}$, and write $r=\text{Re}(a_{(123)})$ and $s=\text{Im}(a_{(123)})$ throughout this section.

Let $\mathcal{P}_0$ be the set of all tuples $(a_e,a_{(12)},a_{(13)},a_{(23)},r,s)$ satisfying \eqref{eq-UUUpos1} and the TP condition
\begin{equation} \label{eq-UUUTP}
    d^2a_e+d(a_{(12)}+a_{(13)}+a_{(23)})+(a_{(123)}+a_{(132)})=1.
\end{equation}
Then $\mathcal{P}_0$ describes the condition $\sum_{\sigma}a_{\sigma}\Le_{\sigma}\in \text{CovPosTP}(\overline{U},UU)$ exactly, so $\mathcal{P}_0$ must be a convex and compact subset of $\Real^6$. For simplification of the condition \eqref{eq-UUUpos1}, let us consider a linear isomorphism
\begin{equation}
    \alpha:(a_e,a_{(12)},a_{(13)},a_{(23)},r,s)\mapsto (a_e,A,B,C,r,s)
\end{equation}
of $\Real^6$, where $A=a_{e}+a_{(12)}$, $B=a_e+a_{(13)}$, and $C=a_{(23)}+r$. Then $\mathcal{P}=\alpha(\mathcal{P}_0)$ is the set of all tuples $(a_e,A,B,C,r,s)\in \Real^6$ satisfying
\begin{equation}\label{eq-UUUpos2}
\left \{\begin{array}{ll}
(1)& A, B\geq 0,\\
(2)& AB\geq C^2+s^2,\\
(3)& A+B+C+r\geq a_e\geq |C-r|,\\
(4)& d(d-2)a_e+d(A+B+C)-(d-2)r=1.
\end{array} \right.
\end{equation}

Note that we have $\text{Ext}(\mathcal{P}_0)=\alpha^{-1}(\text{Ext}(\mathcal{P}))$. That is, it suffices to find the extreme points of $\mathcal{P}$ and restore the coefficients $(a_\sigma)_{\sigma\in \mathcal{S}_3}$ to get the corresponding extremal positive $(\overline{U},UU)$-covariant maps.

\begin{lemma} \label{lem-UUUext2}
Let $\mathcal{S}$ be the set of tuples $(A,B,C,s)$ satisfying (1) and (2) of \eqref{eq-UUUpos2}. Then $\mathcal{S}$ is a convex cone in $\Real^4$. Moreover, if $x_0=x_1+x_2$ in $\mathcal{S}$ with $x_i=(A_i,B_i,C_i,s_i)$ and if $A_0B_0=C_0^2+s_0^2$, then $x_1=\lambda_1 x_0, x_2=\lambda_2 x_0$ for some $\lambda_1,\lambda_2\geq 0$. In other words, the half-line $\Real_{+} x_0$ is an extremal ray of $\mathcal{S}$.
\end{lemma}
\begin{proof}
It is straightforward that $x\in \mathcal{S}$ implies $\lambda x\in \mathcal{S}$ for all $\lambda \geq 0$. Let us write $x_i=(A_i,B_i,C_i,s_i)\in \mathcal{S}$ for $i=1,2$. Then $A_1+A_2\geq 0$ and $B_1+B_2\geq 0$, so the last thing to check is 
\begin{equation}
(A_1+A_2)\cdot (B_1+B_2)\geq (C_1+C_2)^2+(s_1+s_2)^2. 
\end{equation}

Let us choose $C_i'\geq |C_i|$, $s_i'\geq |s_i|$ such that $A_iB_i=(C_i')^2+(s_i')^2$. Then, indeed, we have
\begin{align}
    &(A_1+A_2)(B_1+B_2)\geq A_1B_1+2\sqrt{A_1B_1A_2B_2}+A_2B_2\\
    &=(C_1')^2+(s_1')^2+2\sqrt{((C_1')^2+(s_1')^2)((C_2')^2+(s_2')^2)}+(C_2')^2+(s_2')^2\\
    &\geq(C_1')^2+(s_1')^2+2(C_1'C_2'+s_1's_2')+(C_2')^2+(s_2')^2\\
    &=(C_1'+C_2')^2+(s_1'+s_2')^2\geq (C_1+C_2)^2+(s_1+s_2)^2.
\end{align}
by applying the AM-GM inequality and the Cauchy-Schwarz inequality. Therefore, $x_1+x_2\in \mathcal{S}$, which proves that $\mathcal{S}$ is a convex cone. The latter statement follows by investigating the equality condition carefully in the above inequality, which is left to the reader.
\end{proof}

\begin{proof}[\textbf{Proof of Lemma \ref{lem-UUUext}}]
It is sufficient to show that all the extreme points $\mathbf{x}=(a_e,A,B,C,r,s)$ of $\mathcal{P}$ are classified into the following three types up to normalizing constants: for $A,B\geq 0$ and $AB\geq C^2$,
\[\begin{array}{ll}
\text{Type \MakeUppercase{\romannumeral 1}}'&(1,0,0,0,1,0),\\
\text{Type \MakeUppercase{\romannumeral 2}}'&(0,A,B,C,C,\pm \sqrt{AB-C^2}),\\
\text{Type \MakeUppercase{\romannumeral 3}}'&(\frac{A+B+2C}{2},A,B,C,-\frac{A+B}{2},\pm \sqrt{AB-C^2}).
\end{array}\]

If $AB>C^2+s^2$, then we can choose $s'>|s|$ such that $AB=C^2+(s')^2$. In this case, $\mathbf{x}_{\pm}^{(0)}=(a_e,A,B,C,r,\pm s')\in \mathcal{P}$, and $\mathbf{x}$ is a (nontrivial) convex combination of $\mathbf{x}_+^{(0)}$ and $\mathbf{x}_-^{(0)}$. Thus, $\mathbf{x}$ is not extremal in $\mathcal{P}$.

From now on, we assume $AB=C^2+s^2$ (i.e., $s=\pm \sqrt{AB-C^2}$) and divide the condition (3) of \eqref{eq-UUUpos2} into the following cases.

\textbf{[Case 1]} $A+B+C+r \geq a_e>|C-r|$. Then for sufficiently small $\delta>0$,
\begin{equation}
     \footnotesize \mathbf{x}_{\pm}^{(1)}=\left(a_e\mp \frac{2(A+B+C)}{d-2}\delta, A\pm A\delta, B\pm B\delta, C\pm C\delta, r\mp \frac{d(A+B+C)}{d-2}\delta, s\pm s\delta\right)
\end{equation}
are elements of $\mathcal{P}$, and $\mathbf{x}= (\mathbf{x}_+^{(1)}+\mathbf{x}_-^{(1)})/2$. Therefore, $\mathbf{x}\notin {\rm \text{Ext}}(\mathcal{P})$.

\textbf{[Case 2]} $A+B+C+r>a_e=|C-r|>0$. Here we consider only the case $C>r$ since the other case $C<r$ can be argued similarly. Then for sufficiently small $\delta>0$,
\begin{equation}
     \mathbf{x}_{\pm}^{(2)}=\left(a_e\pm (C-k)\delta, A\pm A\delta, B\pm B\delta, C\pm C\delta, r\pm k\delta, s\pm s\delta\right)
\end{equation}
are elements of $\mathcal{P}$, where $k\in \Real$ satisfies \begin{equation}
    d(d-2)(C-k)+d(A+B+C)-(d-2)k=0
\end{equation}
so that the condition $(4)$ of \eqref{eq-UUUpos2} holds for $\mathbf{x}_{\pm}^{(2)}$. Since $\mathbf{x}= (\mathbf{x}_+^{(2)}+\mathbf{x}_-^{(2)})/2$, it is not extremal.

\textbf{[Case 3]} $A+B+C+r\geq a_e=|C-r|=0$, so $C=r$. We claim that $\mathbf{x}=(0,A,B,C,C,s)\in {\rm \text{Ext}}(\mathcal{P})$ corresponding to $\text{Type \MakeUppercase{\romannumeral 2}}'$. Suppose that $\mathbf{x}$ is a convex combination of $\mathbf{x}_{\pm}^{(3)}=(a_{\pm}, A_{\pm}, B_{\pm},C_{\pm}, r_{\pm}, s_{\pm})\in \mathcal{P}$. Then the condition $a_e=0$ and $a_\pm\geq 0$ imply $a_{\pm}=0$, which again forces $|C_{\pm}-r_{\pm}|=0$. Therefore, Lemma \ref{lem-UUUext2} implies that $\mathbf{x}_{\pm}^{(3)}=(0, A_{\pm}, B_{\pm},C_{\pm}, C_{\pm}, s_{\pm})=\lambda_{\pm}\mathbf{x}$ for some $\lambda_{\pm}\geq 0$. Now the TP condition \eqref{eq-UUUpos2} $(4)$ implies $\lambda_{\pm}=1$, so $\mathbf{x}=\mathbf{x}_{\pm}^{(3)}$.

\textbf{[Case 4]} $A+B+C+r=a_e=C-r\geq 0$. Then $r=-\frac{A+B}{2}$ and $\mathbf{x}=(\frac{A+B+2C}{2}, A,B,C,-\frac{A+B}{2},s)$. Here we claim that $\mathbf{x}\in {\rm \text{Ext}}(\mathcal{P})$ which corresponds to $\text{Type \MakeUppercase{\romannumeral 3}}'$ (note that $A+B\geq 2\sqrt{AB}=2\sqrt{C^2+s^2}\geq 2|C|$, so $r=-\frac{A+B}{2}$ conversely implies $C\geq r$). If $\mathbf{x}$ is a convex combination of $\mathbf{x}_{\pm}^{(4)}=(a_{\pm}, A_{\pm}, B_{\pm},C_{\pm}, r_{\pm}, s_{\pm})\in \mathcal{P}$, then the condition \eqref{eq-UUUpos2} (3) for $\mathbf{x}_\pm^{(4)}$ implies $A_\pm+B_\pm+C_\pm+r_\pm=a_\pm=|C_\pm-r_\pm|$. We may assume $A_+\geq A\geq A_-$ without loss of generality, so Lemma \ref{lem-UUUext2} implies
\begin{equation}
    \footnotesize
    \mathbf{x}_+^{(4)}=\left(\frac{A+B+2C}{2}+\delta', A+A\delta, B+B\delta, C+C\delta, -\frac{A+B}{2}+\delta'',s+s\delta\right) 
\end{equation}
for some $\delta\geq 0$, $\delta',\delta''\in \Real$, and $\delta'=(A+B+C)\delta+\delta''$ from $A_+ +B_+ +C_+ +r_+=a_+$. Now for the case $a_+=r_+-C_+\geq 0$, we have
\begin{equation}
    0\leq A+B+2C=-(A+B+2C)\delta\leq 0.
\end{equation}
However, this says $A+B=-2C$ and $\mathbf{x}=(0,A,B,C,C,s)$, which can be absorbed into \textit{Case 3}. For the case $a_+=C_+-r_+$, we have $\delta''=-\frac{A+B}{2}\delta$ and $\delta'=\frac{A+B+2C}{2}\delta$. However, then the TP condition \eqref{eq-UUUpos2} (4) implies
\begin{equation}
    \left(d(d-2)\frac{A+B+2C}{2}+d(A+B+C) +(d-2)\frac{A+B}{2}\right)\delta=0,
\end{equation}
which is possible only if $\delta=0$. Therefore, $\mathbf{x}=\mathbf{x}_+^{(4)}=\mathbf{x}_-^{(4)}$.

\textbf{[Case 5]} $A+B+C+r=a_e=r-C\geq 0$. Then $A+B=-2C$, and the previous inequality $A+B\geq 2\sqrt{AB}\geq 2|C|$ implies $A=B=-C\geq 0$ and $s=0$. Thus, $\mathbf{x}=(r-C,-C,-C,C,r,0)$ with $C\leq 0$ and $r\geq C$. Then our problem is divided into the following three subcases.
\begin{itemize}
    \item If $C<0$ and $r>C$, then $\mathbf{x}\notin {\rm \text{Ext}}(\mathcal{P})$ since $\mathbf{x}=(\mathbf{x}_+^{(5)}+\mathbf{x}_-^{(5)})/2$, where
    \begin{equation}
        \footnotesize
        \mathbf{x}_{\pm}^{(5)}=\left(r-C\mp\frac{2}{d-2}\delta, -C\pm\delta, -C\pm\delta, C\mp\delta, r\mp\frac{d}{d-2}\delta,0 \right)\in \mathcal{P}
    \end{equation}
    for sufficiently small $\delta>0$.
    
    \item If $r=C$, then $\mathbf{x}=(0,-C,-C,C,C,0)$ is extremal since it can be absorbed into \textit{Case 3}.
    
    \item If $C=0$, then $\mathbf{x}=r(1,0,0,0,1,0)$ is indeed extremal (corresponding to $\text{Type \MakeUppercase{\romannumeral 1}}'$) since the point $(A,B,C,s)=(0,0,0,0)$ is an extreme point of $\mathcal{S}$ in Lemma \ref{lem-UUUext2} and since $r$ is uniquely determined by the TP condition \eqref{eq-UUUpos2} (4).
\end{itemize}
\end{proof}

Now we shall prove Lemma \ref{lem-UUbarUext} using similar arguments. Let $\mathcal{Q}_0$ be the set of all tuples $(a_e,a_{(12)},a_{(13)},a_{(23)},r,s)$ satisfying \eqref{eq-UUbarUpos} and \eqref{eq-UUUTP}, and then consider a linear isomorphism
\begin{equation}
    \beta:(a_e,a_{(12)},a_{(13)},a_{(23)},r,s)\mapsto (A,B,C,p,q,s)
\end{equation}
of $\Real^6$, where $\begin{cases}A=\sum_{\sigma\in \mathcal{S}_3}a_{\sigma}, \; B=\frac{a_e}{d-1}+a_{(23)},\; C=a_{(23)}+r,\\ p=a_e+a_{(12)},\; q=a_e+a_{(13)}.\end{cases}$ Then $\mathcal{Q}=\beta(\mathcal{Q}_0)$ becomes the set of all tuples $(A,B,C,p,q,s)\in \Real^6$ satisfying \small
\begin{equation}\label{eq-UUbarUpos2}
\left \{\begin{array}{ll}
(1)& A, B, p, q\geq 0,\\
(2)& AB\geq C^2+s^2,\\
(3)& A+B-2C\leq p+q,\\
(4)& (-d^2+d+1)A-(d-1)^2 B+2d(d-1)C+(d^2-1)(p+q)=1.
\end{array} \right.
\end{equation}
\normalsize

\begin{proof}[\textbf{Proof of Lemma \ref{lem-UUbarUext}}]
It is sufficient to show that the extreme points $\mathbf{y}=(A,B,C,p,q,s)$ of $\mathcal{Q}$ are classified into the following four types up to normalizing constants: for $A,B\geq 0$ and $AB\geq C^2$,
\[\begin{array}{ll}
\text{Type \MakeUppercase{\romannumeral 1}}'&(0,0,0,1,0,0),\\
\text{Type \MakeUppercase{\romannumeral 2}}'&(0,0,0,0,1,0),\\
\text{Type \MakeUppercase{\romannumeral 3}}'&(A,B,C,A+B-2C,0,\pm \sqrt{AB-C^2}),\\
\text{Type \MakeUppercase{\romannumeral 4}}'&(A,B,C,0,A+B-2C,\pm \sqrt{AB-C^2}).
\end{array}\]
As in the proof of Lemma \ref{lem-UUUext}, we may assume $AB=C^2+s^2$. Furthermore, we may assume $p=0$ or $q=0$ since $\mathbf{y}$ is a convex combination of $\mathbf{y}_\pm^{(0)}\in \mathcal{Q}$, where $\mathbf{y}_+^{(0)}=(A,B,C,p+q,0,s)$ and $\mathbf{y}_-^{(0)}=(A,B,C,0,p+q,s)$. Let us first assume $q=0$ and divide the condition (3) of \eqref{eq-UUbarUpos2} into the following three cases.

\textbf{[Case 1]} $(A,B)\neq (0,0)$ and $A+B-2C < p$. Then for sufficiently small $\delta>0$,
\begin{equation}     
    \mathbf{y}_{\pm}^{(1)}=\left(A\pm A\delta, B\pm B\delta, C\pm C\delta, p\pm \delta',0,s\pm s\delta \right)\in \mathcal{Q}
\end{equation}
where $\delta'\in \Real$ satisfies
\begin{equation}
    \left((-d^2+d+1)A-(d-1)^2B+2d(d-1)C\right)\delta+(d^2-1)\delta'=0,
\end{equation}
so that the condition (4) of \eqref{eq-UUbarUpos2} holds for $\mathbf{y}_\pm^{(1)}$. Since $\mathbf{y}=(\mathbf{y}_+^{(1)}+\mathbf{y}_-^{(1)})/2$ and $\mathbf{y}_+^{(1)}\neq \mathbf{y}_-^{(1)}$, we have $\mathbf{y}\notin \text{Ext}(\mathcal{Q})$.

\textbf{[Case 2]} $A=B=0$ (hence $C=s=0$). Then $\mathbf{y}=p(0,0,0,1,0,0)$ is extremal in $\mathcal{Q}$ (corresponding to $\text{Type \MakeUppercase{\romannumeral 1}}'$) since $(A,B,C,s)=(0,0,0,0)$ is an extreme point of $\mathcal{S}$ in Lemma \ref{lem-UUUext2} and since $p$ is uniquely determined by \eqref{eq-UUbarUpos2} (4).

\textbf{[Case 3]} $A+B-2C=p$. In this case, we claim that $\mathbf{y}=(A,B,C,A+B-2C,0,s)\in \text{Ext}(\mathcal{Q})$, which corresponds to $\text{Type \MakeUppercase{\romannumeral 3}}'$. Indeed, if $\mathbf{y}$ is a convex combination of $\mathbf{y}_\pm^{(2)}=(A_\pm,B_\pm,C_\pm,p_\pm,q_\pm,s_\pm)\in \mathcal{Q}$, then the conditions $q=0$ and $q_\pm\geq 0$ imply $q_\pm=0$. Moreover, the conditions $A+B-2C=p$ and $A_\pm+B_\pm-2C_\pm \leq p_\pm$ imply $A_\pm+B_\pm-2C_\pm = p_\pm$. Now applying Lemma \ref{lem-UUUext2}, we can write
\begin{equation}
    \mathbf{y}_{+}^{(2)}=(A(1+\delta), B(1+\delta), C(1+\delta), (A+B-2C)(1+\delta), 0, s(1+\delta)) 
\end{equation}
for some $\delta\in \Real$. On the other hand, the TP condition \eqref{eq-UUbarUpos2} (4) for $\mathbf{y}_+^{(2)}$ gives
\begin{equation} \label{eq-UUbarUcase2}
    \left(dA+2(d-1)B-2(d-1)C\right)\delta=0.
\end{equation}
However,
\begin{equation}
    dA+2(d-1)B=A+(d-1)B+(d-1)(A+B)\geq 2(d-1)C
\end{equation}
 since $A+B\geq 2C$, and the equality above holds only if $A=B=C=p=s=0$ which is impossible. Therefore, \eqref{eq-UUbarUcase2} holds only if $\delta=0$, and hence we have $\mathbf{y}=\mathbf{y}_+^{(2)}=\mathbf{y}_-^{(2)}$.

Finally, we can proceed analogously when $p=0$ and get the tuples of $\text{Type \MakeUppercase{\romannumeral 2}}'$ and $\text{Type \MakeUppercase{\romannumeral 4}}'$ as extreme points of $\mathcal{Q}$.
\end{proof}

\section{Proof of Theorem \ref{thm-qOOO} when $d=2$} \label{sec-d=2}
When $d=2$, we have an additional relation
\begin{equation}
    V_e-V_{(12)}-V_{(13)}-V_{(23)}+V_{(123)}+V_{(132)}=0.
\end{equation}
Therefore, $\set{V_{\sigma}}_{\sigma\in \mathcal{S}_3}$ is no longer linearly independent, and both the spaces $\text{Inv}(U^{\otimes 3})=\text{span}\set{V_{\sigma}:\sigma\in \mathcal{S}_3}$ and $\text{Inv}(U\otimes \overline{U}\otimes U)=\text{span}\set{T_{\sigma}:\sigma\in \mathcal{S}_3}$ are 5-dimensional. In particular, we have $\text{Inv}(O_+^{\otimes 3})=\text{Inv}(U\otimes \overline{U} \otimes U)$ in this case.

We can write a general element in $\text{Cov}(\overline{U},\overline{U}U)$ as $\mathcal{M}=\sum_{\sigma\in \mathcal{S}_3\backslash \set{e}} a_{\sigma}\mathcal{M}_{\sigma}$. Then $\mathcal{M}$ is positive if and only if
\begin{equation}\label{eq-UUbarUpos3}
\left\{\begin{array}{ll}
a_{(12)}, a_{(13)}, a_{(23)}\geq 0\;\text{ and }\; a_{(132)}=\overline{a_{(123)}},\\
a_{(12)}+a_{(13)}+a_{(23)}+a_{(123)}+a_{(132)}\geq 0,\\
(a_{(12)}+a_{(13)}+a_{(23)}+a_{(123)}+a_{(132)})a_{(23)}\geq|a_{(23)}+a_{(123)}|^2,
\end{array}\right.
\end{equation}
by following the same proof in Lemma \ref{lem-UUbarUpos}. Now let us write $(r,s)=(\text{Re}(a_{(123)}),\text{Im}(a_{(123)}))$ for convenience and consider a linear isomorphism
\begin{equation}
    \tilde{\beta}:(a_{(12)}, a_{(13)},a_{(23)},r,s)\mapsto (a_{(12)},A,B,C,s),
\end{equation}
of $\Real^5$, where $A=a_{(12)}+a_{(13)}+a_{(23)}+2r$, $B=a_{(23)}$, and $C=a_{(23)}+r$. Then the set $\widetilde{\mathcal{Q}}=\set{\tilde{\beta}(a_{(12)},a_{(13)},a_{(23)},r,s):\mathcal{M}\in \text{CovPosTP}(\overline{U},\overline{U}U)}$ is equal to the set of tuples $(a_{(12)},A,B,C,s)\in \Real^5$ satisfying
\begin{equation}\label{eq-UUbarUpos4}
\left\{\begin{array}{ll}
(1)& A,B\geq 0,\\
(2)& AB\geq C^2+s^2,\\
(3)& 0\leq a_{(12)}\leq A+B-2C,\\
(4)& A+B-C=\frac{1}{2}.
\end{array}\right.
\end{equation}

In order to find the extreme points $\mathbf{y}=(a_{(12)},A,B,C,s)$ of $\widetilde{\mathcal{Q}}$, note that we still have $AB=C^2+s^2$ as in the proof of Lemma \ref{lem-UUbarUext}. Moreover, we have $a_{(12)}=0$ or $A+B-2C$ since $\mathbf{y}$ is a convex combination of $\mathbf{y}_+=(A+B-2C,A,B,C,s)$ and $\mathbf{y}_-=(0,A,B,C,s)$. Therefore, we can list all possible extreme points of $\widetilde{\mathcal{Q}}$ in the following two types:
\[\begin{array}{ll}
\text{Type \MakeUppercase{\romannumeral 1}}'& (A+B-2C,A,B,C,\pm \sqrt{AB-C^2}),\\
\text{Type \MakeUppercase{\romannumeral 2}}'& (0,A,B,C,\pm \sqrt{AB-C^2}),
\end{array}\]
for $A,B\geq 0$, $AB\geq C^2$, and $A+B-C=\frac{1}{2}$. Moreover, any extreme point of $\text{Type \MakeUppercase{\romannumeral 1}}'$ corresponds to a tuple \begin{equation}
    (a_{(12)},a_{(13)},a_{(23)},r,s)=(A+B-2C, 0, B,C-B,\pm \sqrt{AB-C^2}),
\end{equation}
so the associated linear map $\mathcal{M}=\sum_{\sigma\in \mathcal{S}_3\setminus \left\{e\right\}}a_{\sigma}\mathcal{M}_{\sigma}$ is CP by Lemma \ref{lem-UUbarUPPT} (note that Lemma \ref{lem-UUbarUPPT} (1) still gives a sufficient condition for $\mathcal{M}$ to be CP when $d=2$). Similarly, any extreme point of $\text{Type \MakeUppercase{\romannumeral 2}}'$ corresponds a CCP map. In other words, every element in $\text{Ext}(\text{CovPosTP}(\overline{U},\overline{U}U))$ is CP or CCP, thus POS=DEC holds in $\text{CovPos}_1(\overline{U}U,\overline{U})$. This completes the proof of Theorem \ref{thm-qOOO} by Corollary \ref{cor-equiv}.

\section{Quantum orthogonal symmetry}\label{sec:FOQG}

Within the framework of \textit{compact quantum groups}, the orthogonal group $\mathcal{O}_d$ is understood as the space $C(\mathcal{O}_d)$ of continuous functions on $\mathcal{O}_d$ endowed with the co-multiplication $\Delta:C(\mathcal{O}_d)\rightarrow C(\mathcal{O}_d\times \mathcal{O}_d)$ given by
\begin{equation}
    (\Delta f)(x,y)=f(xy)
\end{equation}
for all $x,y\in \mathcal{O}_d$ and $f\in C(\mathcal{O}_d)$. %By the Stone-Weierstrass theorem, the $*$-algebra generated by $d^2$ real-valued functions $\pi_{ij}:\mathcal{O}_d\rightarrow \mathbb{R}$ given by $\pi_{ij}(x)=x_{ij}$ for all $x=(x_{ij})_{1\leq i,j\leq d}\in \mathcal{O}_d$ is dense in $C(\mathcal{O}_d)$ in the uniform norm.
Moreover, there exists a family of continuous functions $(\pi_{ij})_{1\leq i,j\leq d}$ generating $C(\mathcal{O}_d)$ and
\begin{equation}
\Delta(\pi_{ij})=\displaystyle \sum_{k=1}^d \pi_{ik}\otimes \pi_{kj}\in C(\mathcal{O}_d)\otimes_{min}C(\mathcal{O}_d)\cong C(\mathcal{O}_d\times \mathcal{O}_d)    
\end{equation}
for all $1\leq i,j\leq d$, where $\otimes_{min}$ means the minimal tensor product between $C^*$-algebras.

The \textit{free orthogonal quantum group} $\mathcal{O}_d^+$ is a liberation of $\mathcal{O}_d$ in the sense that the space $C(\mathcal{O}_d^+)$ of `non-commutative' continuous functions on $\mathcal{O}_d^+$ is the universal unital $C^*$-algebra generated by $d^2$ self-adjoint operators $u_{ij}$ satisfying that $u=\displaystyle \sum_{i,j=1}^d e_{ij}\otimes u_{ij}$ is a unitary, i.e. $u^*u=uu^*=\text{Id}_d\otimes 1$ in $M_{d}(\Comp)\otimes C(\mathcal{O}_d^+)$. The quantum group structure is encoded in the unital $*$-homomorphism $\Delta:C(\mathcal{O}_d^+)\rightarrow C(\mathcal{O}_d^+)\otimes_{min}C(\mathcal{O}_d^+)$ determined by
\begin{equation*}
\Delta(u_{ij})=\sum_{k=1}^d u_{ik}\otimes u_{kj}.    
\end{equation*}
Then $u=\displaystyle \sum_{i,j=1}^d e_{ij}\otimes u_{ij}$ is the standard unitary representation of $\mathcal{O}_d^+$ satisfying $u^c=\displaystyle \sum_{i,j=1}^d e_{ij}\otimes u_{ij}^*=u$ in the sense of \cite{Wo87,Ba96}. The $3$-fold tensor representation $u\tp u\tp u\in M_d(\Comp)^{\otimes 3}\otimes C(\mathcal{O}_d^+)$ of $u$ is defined by 
\begin{equation}
u\tp u\tp u=\displaystyle \sum_{i_1,j_1,i_2,j_2,i_3,j_3=1}^d e_{i_1j_1}\otimes e_{i_2j_2}\otimes e_{i_3j_3}\otimes u_{i_1j_1}u_{i_2j_2}u_{i_3j_3}.    
\end{equation}
Then the space $\text{Inv}(O_+^{\otimes 3})$ in Section \ref{sec:Q.orthogonal} is understood as the space $\text{Inv}(u\tp u \tp u)$ of operators $X\in M_d(\Comp)^{\otimes 3}$ satisfying
\begin{align} \label{eq-qOOOinv}
    (u\tp u \tp u)\cdot (X\otimes 1)=(X\otimes 1)\cdot (u\tp u \tp u)
\end{align}
in view of \cite{LeYo22}. To sketch a proof of this fact, we can observe that the 5 operators $T_{\sigma}$ ($\sigma\in \mathcal{S}_3\backslash \set{(13)}$) in \eqref{eq-OOOinv} are linearly independent, and the operators $T_{\sigma}$ satisfy \eqref{eq-qOOOinv} using the identity
\begin{equation}
    (u\tp u)(|\Om_d\ra\otimes 1)=|\Om_d\ra\otimes 1.
\end{equation}
Thus, $\text{Inv}(O_+^{\otimes 3})\subseteq \text{Inv}(u\tp u \tp u)$. Moreover, the space $\text{Inv}(u\tp u \tp u)$ should be of dimension five thanks to the representation theory of $\mathcal{O}_d^+$ (see Corollary 6.4.12 and Corollary 5.3.5 of \cite{tim08}). Hence, we have $\text{Inv}(O_+^{\otimes 3})= \text{Inv}(u\tp u \tp u)$.

\bibliographystyle{alpha}
\bibliography{PJPY23}

\newcommand{\etalchar}[1]{$^{#1}$}
\def\cprime{$'$}
\begin{thebibliography}{BvDHT99}

\bibitem[AN14]{Al14}
Muneerah Al~Nuwairan.
\newblock The extreme points of {SU}(2)-irreducibly covariant channels.
\newblock {\em Internat. J. Math.}, 25(6):1450048, 30, 2014.

\bibitem[Ban96]{Ba96}
Teodor Banica.
\newblock Th\'{e}orie des repr\'{e}sentations du groupe quantique compact libre
  {${\rm O}(n)$}.
\newblock {\em C. R. Acad. Sci. Paris S\'{e}r. I Math.}, 322(3):241--244, 1996.

\bibitem[BBC{\etalchar{+}}93]{BBC93}
Charles~H. Bennett, Gilles Brassard, Claude Cr\'{e}peau, Richard Jozsa, Asher
  Peres, and William~K. Wootters.
\newblock Teleporting an unknown quantum state via dual classical and
  {E}instein-{P}odolsky-{R}osen channels.
\newblock {\em Phys. Rev. Lett.}, 70(13):1895--1899, 1993.

\bibitem[BCS20]{BCS20}
Ivan Bardet, Beno\^{i}t Collins, and Gunjan Sapra.
\newblock Characterization of equivariant maps and application to entanglement
  detection.
\newblock {\em Ann. Henri Poincar\'{e}}, 21(10):3385--3406, 2020.

\bibitem[Bel64]{Be64}
J.~S. Bell.
\newblock On the {E}instein {P}odolsky {R}osen paradox.
\newblock {\em Phys. Phys. Fiz.}, 1(3):195--200, 1964.

\bibitem[BPM{\etalchar{+}}97]{BPMEWZ}
Dik Bouwmeester, Jian-Wei Pan, Klaus Mattle, Manfred Eibl, Harald Weinfurter,
  and Anton Zeilinger.
\newblock Experimental quantum teleportation.
\newblock {\em Nature}, 390(6660):575--579, 1997.

\bibitem[Bra03]{Br03}
Gilles Brassard.
\newblock Quantum communication complexity.
\newblock volume~33, pages 1593--1616. 2003.
\newblock Special issue dedicated to David Mermin, Part II.

\bibitem[BSST99]{BSSJT}
Charles~H. Bennett, Peter~W. Shor, John~A. Smolin, and Ashish~V. Thapliyal.
\newblock Entanglement-assisted classical capacity of noisy quantum channels.
\newblock {\em Phys. Rev. Lett.}, 83:3081--3084, Oct 1999.

\bibitem[BvDHT99]{BvDHT}
Harry Buhrman, Wim van Dam, Peter H\o{}yer, and Alain Tapp.
\newblock Multiparty quantum communication complexity.
\newblock {\em Phys. Rev. A}, 60:2737--2741, Oct 1999.

\bibitem[BW92]{BeWi92}
Charles~H. Bennett and Stephen~J. Wiesner.
\newblock Communication via one- and two-particle operators on
  {E}instein-{P}odolsky-{R}osen states.
\newblock {\em Phys. Rev. Lett.}, 69(20):2881--2884, 1992.

\bibitem[C{\DJ}13]{CDD}
Lin Chen and Dragomir~{\v{Z}} {\DJ}okovi{\'c}.
\newblock Separability problem for multipartite states of rank at most 4.
\newblock {\em J. Phys. A}, 46(27):275304, 24, 2013.

\bibitem[Cho82]{Ch80}
M.D. Choi.
\newblock Positive linear-maps.
\newblock In {\em in Operator Algebras and Applications (Kingston, 1980),
  Proc.Sympos.Pure.Math.}, volume~38, pages 583--590. Amer.Math.Soc., 1982.

\bibitem[CKK{\etalchar{+}}21]{CKKLY21}
Euijung Chang, Jaeyoung Kim, Hyesun Kwak, Hun~Hee Lee, and Sang-Gyun Youn.
\newblock Irreducibly $ su (2) $-covariant quantum channels of low rank.
\newblock {\em to appear in Rev. Math. Phys., arXiv preprint arXiv:2105.00709},
  2021.

\bibitem[COS18]{COS18}
Beno\^{i}t Collins, Hiroyuki Osaka, and Gunjan Sapra.
\newblock On a family of linear maps from {$M_n(\Bbb C)$} to {$M_{n^2}(\Bbb
  C)$}.
\newblock {\em Linear Algebra Appl.}, 555:398--411, 2018.

\bibitem[DPR07]{DPR07b}
A.~R.~Usha Devi, R.~Prabhu, and A.~K. Rajagopal.
\newblock Characterizing multiparticle entanglement in symmetric $n$-qubit
  states via negativity of covariance matrices.
\newblock {\em Phys. Rev. Lett.}, 98:060501, Feb 2007.

\bibitem[EK00]{EK00}
Myoung-Hoe Eom and Seung-Hyeok Kye.
\newblock Duality for positive linear maps in matrix algebras.
\newblock {\em Math. Scand.}, 86(1):130--142, 2000.

\bibitem[Eke91]{Ek91}
Artur~K. Ekert.
\newblock Quantum cryptography based on {B}ell's theorem.
\newblock {\em Phys. Rev. Lett.}, 67(6):661--663, 1991.

\bibitem[EW01]{EW01}
T.~Eggeling and R.~F. Werner.
\newblock Separability properties of tripartite states with
  ${U}\ensuremath{\bigotimes}{U}\ensuremath{\bigotimes}{U}$ symmetry.
\newblock {\em Phys. Rev. A}, 63:042111, Mar 2001.

\bibitem[G\"11]{Gu11}
Otfried G\"{u}hne.
\newblock Entanglement criteria and full separability of multi-qubit quantum
  states.
\newblock {\em Phys. Lett. A}, 375(3):406--410, 2011.

\bibitem[GBW21]{GBW21}
Martina Gschwendtner, Andreas Bluhm, and Andreas Winter.
\newblock Programmability of covariant quantum channels.
\newblock {\em Quant.}, 5(6):488, 2021.

\bibitem[Gha10]{Gh10}
Sevag Gharibian.
\newblock Strong {NP}-hardness of the quantum separability problem.
\newblock {\em Quantum Inf. Comput.}, 10(3-4):343--360, 2010.

\bibitem[Gur03]{Gu03}
Leonid Gurvits.
\newblock Classical deterministic complexity of {E}dmond's problem and quantum
  entanglement.
\newblock In {\em Proceedings of the {T}hirty-{F}ifth {A}nnual {ACM}
  {S}ymposium on {T}heory of {C}omputing}, pages 10--19. ACM, New York, 2003.

\bibitem[HH99]{HoHo99}
Micha\l{} Horodecki and Pawe\l{} Horodecki.
\newblock Reduction criterion of separability and limits for a class of
  distillation protocols.
\newblock {\em Phys. Rev. A}, 59:4206--4216, Jun 1999.

\bibitem[HHH96]{HHH96}
Micha\l Horodecki, Pawe\l Horodecki, and Ryszard Horodecki.
\newblock Separability of mixed states: necessary and sufficient conditions.
\newblock {\em Phys. Lett. A}, 223(1-2):1--8, 1996.

\bibitem[HHH99]{HHH99}
Pawe\l{} Horodecki, Micha\l{} Horodecki, and Ryszard Horodecki.
\newblock Bound entanglement can be activated.
\newblock {\em Phys. Rev. Lett.}, 82:1056--1059, Feb 1999.

\bibitem[HHHO05]{HHHO05}
Karol Horodecki, Micha\l{} Horodecki, Pawe\l{} Horodecki, and Jonathan
  Oppenheim.
\newblock Secure key from bound entanglement.
\newblock {\em Phys. Rev. Lett.}, 94:160502, Apr 2005.

\bibitem[HHHO09]{HHHO09}
Karol Horodecki, Micha\l Horodecki, Pawe\l Horodecki, and Jonathan Oppenheim.
\newblock General paradigm for distilling classical key from quantum states.
\newblock {\em IEEE Trans. Inform. Theory}, 55(4):1898--1929, 2009.

\bibitem[HK16a]{HK16}
Kil-Chan Ha and Seung-Hyeok Kye.
\newblock Construction of three-qubit genuine entanglement with bipartite
  positive partial transposes.
\newblock {\em Phys. Rev. A}, 93:032315, Mar 2016.

\bibitem[HK16b]{HaKy16}
Kyung~Hoon Han and Seung-Hyeok Kye.
\newblock Construction of multi-qubit optimal genuine entanglement witnesses.
\newblock {\em J. Phys. A}, 49(17):17503, 16, 2016.

\bibitem[HPHH08]{HPHH08}
Karol Horodecki, {\L}ukasz Pankowski, Micha{\l} Horodecki, and Pawe{\l}
  Horodecki.
\newblock Low-dimensional bound entanglement with one-way distillable
  cryptographic key.
\newblock {\em IEEE Trans. Inform. Theory}, 54(6):2621--2625, 2008.

\bibitem[HSR03]{HSR03}
Michael Horodecki, Peter~W. Shor, and Mary~Beth Ruskai.
\newblock Entanglement breaking channels.
\newblock {\em Rev. Math. Phys.}, 15(6):629--641, 2003.

\bibitem[JSW{\etalchar{+}}00]{JSWWZ00}
Thomas Jennewein, Christoph Simon, Gregor Weihs, Harald Weinfurter, and Anton
  Zeilinger.
\newblock Quantum cryptography with entangled photons.
\newblock {\em Phys. Rev. Lett.}, 84:4729--4732, May 2000.

\bibitem[Kay11]{Ka11}
Alastair Kay.
\newblock Optimal detection of entanglement in greenberger-horne-zeilinger
  states.
\newblock {\em Phys. Rev. A}, 83:020303, Feb 2011.

\bibitem[KCL05]{KCL05}
J.~K. Korbicz, J.~I. Cirac, and M.~Lewenstein.
\newblock Spin squeezing inequalities and entanglement of $n$ qubit states.
\newblock {\em Phys. Rev. Lett.}, 95:120502, Sep 2005.

\bibitem[KMS20]{KMS20}
Piotr Kopszak, Marek Mozrzymas, and Micha\l Studzi\'{n}ski.
\newblock Positive maps from irreducibly covariant operators.
\newblock {\em J. Phys. A}, 53(39):395306, 33, 2020.

\bibitem[Kye23]{Kye23}
Seung-Hyeok Kye.
\newblock Compositions and tensor products of linear maps between matrix
  algebras.
\newblock {\em Linear Algebra Appl.}, 658:283--309, 2023.

\bibitem[Las10]{La10}
Jean~Bernard Lasserre.
\newblock {\em Moments, positive polynomials and their applications}, volume~1
  of {\em Imperial College Press Optimization Series}.
\newblock Imperial College Press, London, 2010.

\bibitem[LY22]{LeYo22}
Hun~Hee Lee and Sang-Gyun Youn.
\newblock Quantum channels with quantum group symmetry.
\newblock {\em Comm. Math. Phys.}, 389(3):1303--1329, 2022.

\bibitem[Mas06]{Ma06}
Llu\'{\i}s Masanes.
\newblock All bipartite entangled states are useful for information processing.
\newblock {\em Phys. Rev. Lett.}, 96:150501, Apr 2006.

\bibitem[MRS15]{MRS15}
Marek Mozrzymas, Adam Rutkowski, and Micha\l Studzi\'{n}ski.
\newblock Using non-positive maps to characterize entanglement witnesses.
\newblock {\em J. Phys. A}, 48(39):395302, 11, 2015.

\bibitem[MWKZ96]{MWKZ96}
Klaus Mattle, Harald Weinfurter, Paul~G. Kwiat, and Anton Zeilinger.
\newblock Dense coding in experimental quantum communication.
\newblock {\em Phys. Rev. Lett.}, 76:4656--4659, Jun 1996.

\bibitem[NPW{\etalchar{+}}00]{NPWBK00}
D.~S. Naik, C.~G. Peterson, A.~G. White, A.~J. Berglund, and P.~G. Kwiat.
\newblock Entangled state quantum cryptography: Eavesdropping on the ekert
  protocol.
\newblock {\em Phys. Rev. Lett.}, 84:4733--4736, May 2000.

\bibitem[NZ16]{NiZh16}
Jiawang Nie and Xinzhen Zhang.
\newblock Positive maps and separable matrices.
\newblock {\em SIAM J. Optim.}, 26(2):1236--1256, 2016.

\bibitem[Pau02]{Pau02}
Vern Paulsen.
\newblock {\em Completely bounded maps and operator algebras}, volume~78 of
  {\em Cambridge Studies in Advanced Mathematics}.
\newblock Cambridge University Press, Cambridge, 2002.

\bibitem[Per96]{Pe96}
Asher Peres.
\newblock Separability criterion for density matrices.
\newblock {\em Phys. Rev. Lett.}, 77(8):1413--1415, 1996.

\bibitem[SDN22]{SDN22}
Satvik Singh, Nilanjana Datta, and Ion Nechita.
\newblock Ergodic theory of diagonal orthogonal covariant quantum channels,
  2022.

\bibitem[SN21]{SiNe21}
Satvik Singh and Ion Nechita.
\newblock Diagonal unitary and orthogonal symmetries in quantum theory.
\newblock {\em {Quantum}}, 5:519, August 2021.

\bibitem[SN22]{SiNe22b}
Satvik Singh and Ion Nechita.
\newblock {The PPT$^2$ Conjecture Holds for All Choi-Type Maps}.
\newblock {\em Annales Henri Poincare}, 23(9):3311--3329, 2022.

\bibitem[Sr82]{Sto82}
Erling St\o~rmer.
\newblock Decomposable positive maps on {$C^{\ast} $}-algebras.
\newblock {\em Proc. Amer. Math. Soc.}, 86(3):402--404, 1982.

\bibitem[TBZG00]{TBZG00}
W.~Tittel, J.~Brendel, H.~Zbinden, and N.~Gisin.
\newblock Quantum cryptography using entangled photons in energy-time bell
  states.
\newblock {\em Phys. Rev. Lett.}, 84:4737--4740, May 2000.

\bibitem[TG09]{TG09}
G\'eza T\'oth and Otfried G\"uhne.
\newblock Entanglement and permutational symmetry.
\newblock {\em Phys. Rev. Lett.}, 102:170503, May 2009.

\bibitem[Tim08]{tim08}
Thomas Timmermann.
\newblock {\em An Invitation to Quantum Groups and Duality}.
\newblock European Mathematical Society, 2008.

\bibitem[VW01]{VW01}
K.~G.~H. Vollbrecht and R.~F. Werner.
\newblock Entanglement measures under symmetry.
\newblock {\em Phys. Rev. A}, 64:062307, Nov 2001.

\bibitem[VW02]{VW02}
Karl Gerd~H. Vollbrecht and Michael~M. Wolf.
\newblock Activating distillation with an infinitesimal amount of bound
  entanglement.
\newblock {\em Phys. Rev. Lett.}, 88:247901, May 2002.

\bibitem[Wan95]{Wa95}
Shuzhou Wang.
\newblock Free products of compact quantum groups.
\newblock {\em Comm. Math. Phys.}, 167(3):671--692, 1995.

\bibitem[Wat18]{wat2018}
John Watrous.
\newblock {\em The Theory of Quantum Information}.
\newblock Cambridge University Press, 2018.

\bibitem[Wer89]{We89}
Reinhard~F. Werner.
\newblock Quantum states with einstein-podolsky-rosen correlations admitting a
  hidden-variable model.
\newblock {\em Phys. Rev. A}, 40:4277--4281, Oct 1989.

\bibitem[Wor76]{Wo76}
S.~L. Woronowicz.
\newblock Nonextendible positive maps.
\newblock {\em Comm. Math. Phys.}, 51(3):243--282, 1976.

\bibitem[Wor87]{Wo87}
S.~L. Woronowicz.
\newblock Compact matrix pseudogroups.
\newblock {\em Comm. Math. Phys.}, 111(4):613--665, 1987.

\end{thebibliography}

\end{document}